\documentclass[11pt]{article}
\usepackage[margin=1in]{geometry}
\usepackage[T1]{fontenc}
\usepackage[utf8]{inputenc}
\usepackage{lmodern}
\usepackage{mathtools,amssymb,bm,amsthm,mathrsfs}
\newtheorem{proposition}{Proposition}[section]
\newtheorem{lemma}[proposition]{Lemma}
\newtheorem{theorem}[proposition]{Theorem}
\newtheorem{corollary}[proposition]{Corollary}
\usepackage{booktabs,array}
\usepackage{microtype}
\usepackage{enumitem}
\usepackage{xcolor}
\usepackage{graphicx}
\usepackage{tikz}
\usetikzlibrary{arrows.meta,positioning,shapes.geometric,fit,calc,decorations.pathreplacing}
\usepackage[hidelinks]{hyperref}
\usepackage[nameinlink,noabbrev]{cleveref}

\newcommand{\R}{\mathbb{R}}
\newcommand{\dd}{\,\mathrm{d}}
\newcommand{\vct}[1]{\bm{#1}}

\newcommand{\tr}{\operatorname{tr}}
\newcommand{\dev}{\operatorname{dev}}

\newcommand{\sym}{\operatorname{sym}}

\newcommand{\Conf}{\mathrm{conf}}

\newcommand{\ec}{\mathrm{ec}}

\renewcommand{\arraystretch}{1.14}

\begin{document}

\title{From a three-fluid junction to Young's law:\\
Constitutive and rigid-support limits of wetting}
\author{Dieter Bothe\\[0.4em]
\small Fachbereich Mathematik and Profile Topic Thermo-Fluids and Interfacial Phenomena,\\
\small Technische Universit\"at Darmstadt, Peter-Gr\"unberg-Str.~10, 64287 Darmstadt, Germany}
\date{15 September 2026}

\maketitle
\begin{abstract}
Young's law on a rigid solid and the Neumann balance at a three-fluid junction describe different admissible motions. We connect them through an ordered continuum construction. An affine viscoelastic third phase first acquires permanent deformation memory in the no-relaxation limit at fixed modulus. Independently specified material-surface laws distinguish surface free energy from mechanical surface stress. On the deformable support, displacement of solid material and migration of the wet--dry partition give separate equilibrium conditions. Stiffening an anchored support suppresses displacement at a fixed macroscopic observation distance but leaves reversible migration possible. The surviving variation gives Young's law; the constrained mechanics remains as a reaction carrying normal and, in general, tangential capillary loads. Two complementary results substantiate this reduction. A clamped finite-strain model with circular liquid interfaces yields uniform convergence of reduced energies and recovery of the spatial Young angle for global almost minimisers. A separate quadratic half-space model gives convergence of stationary migration derivatives and a nonzero limiting traction distribution. The results explain how a local Neumann-like ridge can coexist with a Young angle measured against the distant wall, while distinguishing geometric convergence from the control of migration and reaction forces.
\end{abstract}


\section{Introduction}
\label{sec:intro}

For a liquid $L$ on a homogeneous rigid solid $S$ in a gas $G$, Young's law \cite{Young1805} reads
\begin{equation}
\gamma_{SG}-\gamma_{SL}=\gamma_{LG}\cos\theta_Y.
\label{eq:young-intro}
\end{equation}
Here $\gamma_{SG}$ and $\gamma_{SL}$ are the equilibrium solid--gas and solid--liquid surface potentials per current area, $\gamma_{LG}$ is the liquid--gas tension, and $\theta_Y$ is the equilibrium angle measured through the liquid \cite{Young1805}. At a freely deformable three-fluid junction, the corresponding capillary construction is the vector Neumann balance. Their connection requires identifying how the forces and admissible motions change when one phase acquires an elastic constitutive response and the resulting solid is made rigid.

A deformable support makes explicit the mechanics suppressed by the prescribed geometry of a rigid wall. Capillary loading creates a wetting ridge, a localised deformation of the solid surface at the contact line, whose geometry reflects solid surface stresses \cite{Lester1961,Jerison2011}. Those stresses need not equal the solid surface free energies entering Young's law. Moreover, moving solid material and advancing the boundary between its wet and dry regions are different operations. A rigid constraint suppresses the first but can leave the second admissible. The capillary load survives through the support reaction, while migration selects the equilibrium contact angle.

The distinction between spatial and material variations is established. Olives derives separate conditions for a line fixed on a deformable body and for motion relative to that body \cite{Olives2010}; related mechanical and configurational conditions occur in membrane and rigid-substrate theories \cite{HuiJagota2015,LiuEtAl2020Membrane,FriedJabbour2012}. Snoeijer, Rolley and Andreotti identify the kinematic linkage of spatial and material line motion on a rigid substrate \cite{SnoeijerRolleyAndreotti2018}. Nonlinear ridge theory further clarifies the roles of surface elasticity, weak elastic singularities and pinning \cite{Pandey2020}. The affine viscoelastic--elastic correspondence\footnote{Affine transport means that the stored configuration is stretched and rotated by the same local deformation gradient as the surrounding material.} is also known independently of wetting \cite{SnoeijerEtAl2020Elasticity}.

Here we organise these ingredients into a route from a three-fluid junction through an affine viscoelastic third phase and an elastic support to the rigid-wall problem (Figure~\ref{fig:genealogy}).
The bulk constitutive extension, the surface laws and the external anchoring are separate modelling inputs. Thus this is an ordered construction, not a single-parameter limit of one parent bulk--surface functional or a model of chemical solidification. In particular, the final solid surface energies are specified independently; the initial three-fluid tensions and the bulk relaxation time do not determine the Young angle.

The new results control two distinct aspects of the stiff-support limit: recovery of the spatial Young angle from global minimality, and convergence of stationary migration derivatives together with the retained reaction. A reduced energy depends on the wet--dry partition after the other fields have been eliminated; the finite-strain result uses their energy infimum. Global almost minimisers have energy within a vanishing tolerance of the global infimum. Appendix~\ref{sec:nonlinear-continuum} proves uniform convergence of reduced energies and Young-angle recovery for global almost minimisers in a clamped finite-strain model with circular liquid interfaces. Appendix~\ref{sec:quadratic-continuum} solves a separate quadratic half-space model, controlling stationary migration derivatives and the limiting traction distribution. The main text explains the physical reduction and the assumptions of each model; the appendices give the detailed estimates, calculations and proofs. The two continuum results have complementary scopes and are not joined by an approximation theorem.

We consider isothermal reversible partial wetting of a non-swelling solid, with local material-surface energies, specified homogeneous equilibrium surface states, fixed liquid volume, and no retained line energy or pinning barrier. The support is anchored independently of the migrating partition. The limiting angle is measured against the prescribed wall away from the shrinking deformation region. We call that reference geometry the \emph{outer wall}; the observation length relative to the deformation region is specified in Section~\ref{sec:stiff}. Section~\ref{sec:solid-surface} specifies the surface-state thermodynamics, and Section~\ref{sec:landscape} compares neighbouring support and force-transmission models. Technical derivations are collected in Appendix~\ref{sec:technical-details}.

\begin{figure}[tbp]
\centering
\resizebox{0.98\textwidth}{!}{%
\begin{tikzpicture}[
  >=Latex,
  font=\scriptsize,
  state/.style={draw,rounded corners,align=center,minimum height=11mm,inner sep=2.0mm},
  flow/.style={->,line width=0.65pt}
]
\node[font=\bfseries,anchor=east] at (-2.15,2.00) {bulk};
\node[state,minimum width=27mm] (f3) at (0,2.00) {three simple\\fluid bulks};
\node[state,minimum width=31mm] (ve) at (4.15,2.00) {affine viscoelastic\\third phase};
\node[state,minimum width=28mm] (es) at (8.25,2.00) {elastic bulk\\solid $S$};
\node[state,minimum width=29mm] (rig) at (12.25,2.00) {stiff / rigid\\bulk support};
\draw[flow] (f3.east)--(ve.west);
\draw[flow] (ve.east)--(es.west);
\draw[flow] (es.east)--(rig.west);
\node[align=center] at ($(f3.east)!0.5!(ve.west)+(0,1.02)$) {bulk state-space\\extension};
\node[align=center] at ($(ve.east)!0.5!(es.west)+(0,1.02)$) {no relaxation\\fixed modulus};
\node[align=center] at ($(es.east)!0.5!(rig.west)+(0,1.02)$) {stiffening\\fixed reference geometry};

\node[font=\bfseries,anchor=east] at (-2.15,-0.65) {surface};
\node[state,minimum width=34mm] (fi) at (0,-0.65) {fluid-like\\surface laws};
\node[state,minimum width=42mm] (ss) at (8.25,-0.65) {material solid\\surface laws};
\node[state,minimum width=32mm] (eq) at (12.25,-0.65) {fixed-wall\\surface states};
\draw[flow] (fi.east)--(ss.west);
\draw[flow] (ss.east)--(eq.west);
\node[align=center] at ($(fi.east)!0.5!(ss.west)+(0,0.96)$) {independent surface constitutive extension};
\node[align=center] at ($(ss.east)!0.5!(eq.west)+(0,0.96)$) {constraint reduction};

\node[font=\bfseries,anchor=east] at (-2.15,-2.85) {balance};
\node[align=center] at (0,-2.85) {Neumann junction:\\spatial force balance};
\node[align=center] at (8.25,-2.85) {deformable solid:\\spatial and material\\variations};
\node[align=center] at (12.25,-2.85) {rigid-wall model:\\support reaction\\and Young};
\end{tikzpicture}%
}
\caption{Two coordinated constitutive choices. The bulk route introduces permanent deformation memory and then increases stiffness under fixed anchoring. The surface route independently specifies the material surface laws and their prescribed-wall states. Horizontal alignment indicates where the two descriptions are used together; the bulk limit does not determine the surface law.}
\label{fig:genealogy}
\end{figure}
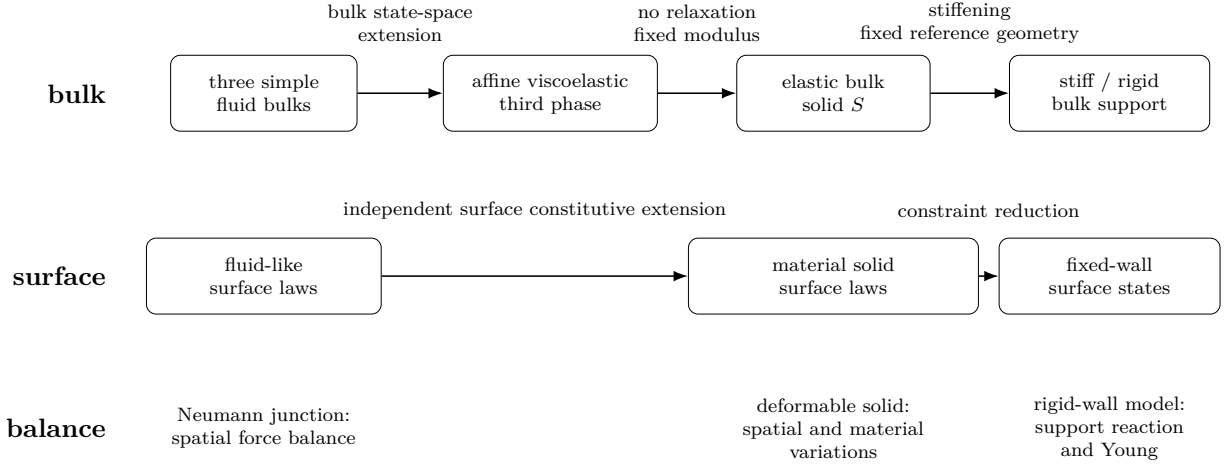

\section{The three-fluid starting point}
\label{sec:threefluid}

Let three bulk regions $\Omega_i$, $i=1,2,3$, meet along a smooth contact line $\Gamma$, with pairwise interfaces $\Sigma_{ij}$. At the line, $\vct N_{ij}$ denotes the unit co-normal tangent to the interface, normal to $\Gamma$, and pointing \emph{from the line into the interface sheet}. A tensile interface therefore pulls the line in the direction $+\vct N_{ij}$. This convention is used throughout.

For a simple isotropic fluid interface at equilibrium,
\begin{equation}
\vct S^{\Sigma_{ij}}=\gamma_{ij}\vct P_{ij},
\qquad \vct P_{ij}=\vct I-\vct n_{ij}\otimes\vct n_{ij}.
\label{eq:fluid-surface-stress}
\end{equation}
Here $\vct S^{\Sigma_{ij}}$ is the mechanical surface-stress tensor, $\vct n_{ij}$ is a unit interface normal, and $\vct P_{ij}$ projects onto the interface tangent plane. The scalar $\gamma_{ij}$ is the equilibrium area-conversion potential, meaning the potential conjugate to reversible area change at the imposed thermodynamic controls, as well as the isotropic surface stress.\footnote{Bold symbols denote vectors or tensors; $\vct I$ is the identity tensor. Superscript $\mathsf T$ denotes transpose and $\otimes$ the tensor product. For a matrix $\vct M$, we use $\tr$ for trace, $\sym\vct M=(\vct M+\vct M^{\mathsf T})/2$, and $\dev\vct M=\vct M-(\tr\vct M)\vct I/d$ in $d$ spatial dimensions. A colon denotes the tensor inner product.} With no retained line force or line energy, and no concentrated bulk force at the junction, line equilibrium gives
\begin{equation}
\gamma_{12}\vct N_{12}+\gamma_{23}\vct N_{23}
+\gamma_{31}\vct N_{31}=\vct0.
\label{eq:neumann}
\end{equation}
A nondegenerate Neumann junction requires each positive tension to be smaller than the sum of the other two, the strict triangle inequalities. The bulk--surface--line balance laws are discussed in Refs.~\cite{BothePruss2016,SlatterySagisOh2007}; only their reversible equilibrium specialisation is needed here.

The same condition follows from varying the interfacial energy $\mathcal E_\Sigma=\sum_{i<j}\gamma_{ij}|\Sigma_{ij}|$, where $|\Sigma_{ij}|$ denotes current area. A spatial line displacement $\delta\vct x_\Gamma$ produces
\begin{equation}
\delta\mathcal E_\Sigma\big|_\Gamma
=-\left(\sum_{i<j}\gamma_{ij}\vct N_{ij}\right)
\cdot\delta\vct x_\Gamma
\label{eq:neumann-variation}
\end{equation}
per unit line length, after the interior constrained equilibrium equations have been imposed. Thus one spatial variation gives the mechanical Neumann balance. Once a material solid surface is introduced, the energy also distinguishes deformation of existing surface material from conversion between wet and dry states.

\section{Permanent bulk memory: the constitutive solid limit}
\label{sec:viscoelastic-solid}

A Newtonian stress $-p\vct I+2\eta\vct D$, with pressure $p$, viscosity $\eta$ and symmetric velocity gradient $\vct D$, stores no static shear energy. Sending $\eta$ to infinity may retain a constraint reaction through a vanishing strain rate, but it does not generate an elastic reference response. We instead give phase $3$ a symmetric positive-definite conformation tensor $\vct A$, representing its molecular configuration and transported affinely by the bulk motion. In the incompressible, isothermal Oldroyd-B class,
\begin{align}
\vct T_3&=-p_3\vct I+2\eta_s\vct D_3+G_S(\vct A-\vct I),
\label{eq:oldroyd-stress}\\
\overset{\triangledown}{\vct A}
&:=\frac{D\vct A}{Dt}-\vct L\vct A-\vct A\vct L^{\mathsf T}
=-\lambda^{-1}(\vct A-\vct I),
\qquad L_{ij}=\partial_j(v_3)_i.
\label{eq:oldroyd-A}
\end{align}
Here $\vct T_3$ is the bulk Cauchy stress, $p_3$ the incompressibility multiplier, $\vct v_3$ the material velocity, $\vct L$ its spatial gradient, and $\vct D_3=\sym\vct L$. The derivative $D/Dt=\partial_t+\vct v_3\cdot\nabla$ follows a material trajectory, where $t$ is time, $\nabla$ is the current spatial gradient and $\partial_j$ differentiates with respect to its $j$th spatial coordinate. The constants $G_S>0$, $\lambda>0$ and $\eta_s\ge0$ are the elastic modulus, relaxation time and parallel viscosity; the polymeric viscosity is $\eta_p=G_S\lambda$. Constituent $3$ undergoes no phase transfer across its interfaces in this passage. A broader two-fluid thermodynamic formulation is given in Ref.~\cite{Bothe2026Polymer}; the present argument begins with the restricted constitutive equations above.

For the same material motion, let $\vct F$ be the deformation gradient from reference to current coordinates. It obeys $D\vct F/Dt=\vct L\vct F$. Hence the left Cauchy--Green tensor $\vct B=\vct F\vct F^{\mathsf T}$ satisfies $\overset{\triangledown}{\vct B}=\vct0$. Take
\begin{equation}
\lambda\to\infty,
\qquad G_S=\eta_p/\lambda=\text{fixed}>0.
\label{eq:lambda-limit}
\end{equation}
The limiting conformation is no longer relaxed. With matched stress-free data, $\vct A(0)=\vct B(0)=\vct I$, uniqueness of the transport equation gives
\begin{equation}
\vct A=\vct B,
\qquad
\vct T_3=-p_3\vct I+2\eta_s\vct D_3+G_S(\vct B-\vct I).
\label{eq:A-equals-B}
\end{equation}
Its elastic deviatoric part is $G_S\dev\vct B$. The isotropic difference from the usual incompressible neo-Hookean stress is absorbed into pressure. With $\eta_s>0$, a viscous contribution remains in parallel with the elastic response \cite{SnoeijerEtAl2020Elasticity}. Holding $\eta_p$ fixed instead would give $G_S\to0$ and lose this finite-modulus solid.

Material labels already exist for a fluid, but a simple-fluid stress law does not retain static shear relative to those labels. The no-relaxation limit makes that reference information mechanically observable: the conformation stores deformation indefinitely, and the stress responds to it even after motion ceases. Only phase $3$ undergoes this constitutive change; phases $1=L$ and $2=G$ retain their fluid laws. The junction geometry and balance laws remain in place, while the state variables and constitutive response of the third phase change.

The conformation energy density $W_p$ and its limiting neo-Hookean density $W_{NH}$ confirm the correspondence:
\begin{equation}
W_p(\vct A)=\frac{G_S}{2}
\left(\tr\vct A-\ln\det\vct A-3\right)
\ \longrightarrow\
W_{NH}(\vct B)=\frac{G_S}{2}(\tr\vct B-3),
\label{eq:oldroyd-energy}
\end{equation}
because $\det\vct B=1$ for matched incompressible memory. The logarithmic term is needed on the full positive-definite conformation space at finite $\lambda$ \cite{BoyavalEtAl2009}. Appendix~\ref{sec:transport-estimate} gives the dissipation identity and an error estimate bounded by a constant times $T/\lambda$ along a prescribed common motion on an observation interval of fixed duration $T$. General initial memory takes the form $\vct A=\vct F\vct D_0\vct F^{\mathsf T}$, where $\vct D_0$ is a time-independent referential tensor field fixed by the initial state; a compatible stress-free reference is an additional property of $\vct D_0$, not a consequence of its positivity.

Mechanical stretching supplies both stored conformation energy and relaxation losses. The energy balance separates two dissipative mechanisms. Relaxation reduces the conformation energy toward its isotropic minimum, whereas the parallel viscous stress dissipates energy whenever the solid moves. At fixed modulus, the former contribution vanishes in the no-relaxation limit on a bounded observation interval with bounded conformation and inverse conformation. The latter can remain and damp motion about an elastic equilibrium. Thus permanent elastic memory is compatible with viscous damping; removing relaxation does not require removing every source of dissipation. The exact nonnegative dissipation densities are derived in Appendix~\ref{sec:transport-estimate}.

The order relative to long time matters. At finite $\lambda$, a quiescent conformation relaxes to $\vct I$. With $\lambda=\infty$ imposed first, deformation memory persists and can support an elastic equilibrium. We study stationarity of that elastic problem; we do not assume a uniform-in-time limit of the original relaxing liquid. Without damping, dynamical convergence to a static state is not automatic.

The correspondence also depends on the constitutive transport. The Deborah number is the ratio of relaxation time to a chosen observation time. A large value means that relaxation is slow on that interval, but it does not by itself establish a permanent elastic reference. Moreover, objective transport laws, which are invariant under a change of observer, need not describe the same deformation memory. The non-affine Gordon--Schowalter/Johnson--Segalman class permits the stored configuration to deform differently from the bulk. Away from its upper- and lower-convected endpoints, its no-relaxation transport generally does not determine a unique elastic state from the total deformation alone \cite{SnoeijerEtAl2020Elasticity}. Frame indifference alone does not determine a transport law independently of the chosen state variable, free energy and stress. The upper-convected endpoint used here transports conformation with the bulk deformation. The lower-convected endpoint instead has an inverse-deformation correspondence and does not describe the same physical affine transport. Our construction relies specifically on \eqref{eq:oldroyd-A}.

These observations fix the order of the construction. Increasing the modulus while retaining finite relaxation can produce a very stiff response over short times and still allow stress relaxation over long times. We first select the permanent-memory elastic model and then use its modulus to suppress capillary deformation. The bulk construction leaves the surface constitutive laws to be specified independently.

\section{Solid surface energy and surface stress}
\label{sec:solid-surface}

Let $\mathcal B_S$ be the reference solid body, with placement $\vct\chi_S$ mapping reference points to their current positions, and material surface $\mathcal S=\partial\mathcal B_S$. A coherently bonded surface has deformation gradient $\vct F_s$, the tangential restriction of the bulk gradient, and current-to-reference area ratio $J_s$. Its local thermodynamic potential per reference area is $\Psi_{S\alpha}(\vct F_s,\vct q_{S\alpha})$, where $\alpha=L,G$ and $\vct q_{S\alpha}$ denotes any additional specified or equilibrated surface state. At fixed thermodynamic controls,
\begin{equation}
\gamma_{S\alpha}=\frac{\Psi_{S\alpha}}{J_s},
\qquad
\vct P_{s,S\alpha}=\frac{\partial\Psi_{S\alpha}}{\partial\vct F_s},
\qquad
\vct\Upsilon_{S\alpha}=J_s^{-1}\vct P_{s,S\alpha}\vct F_s^{\mathsf T}.
\label{eq:surface-piola-pushforward}
\end{equation}
The scalar $\gamma_{S\alpha}$ is the equilibrium area-conversion potential, $\vct P_{s,S\alpha}$ the surface first Piola stress, and $\vct\Upsilon_{S\alpha}$ the current mechanical surface stress. Their distinction is the Shuttleworth effect \cite{Shuttleworth1950,GurtinMurdoch1975,AndreottiSnoeijer2016}. Appendix~\ref{sec:surface-stress-exact} gives the finite-strain derivative with explicit measures.

The special law $\Psi_{S\alpha}=\gamma_{S\alpha}J_s$, with strain-independent $\gamma_{S\alpha}$, gives $\vct\Upsilon_{S\alpha}=\gamma_{S\alpha}\vct P_{\Sigma_s}$, where $\vct P_{\Sigma_s}$ is the current surface tangent projector. Thus a liquid-like surface response is included as a constitutive choice.

The physical distinction can be expressed through two operations on a marked patch of solid surface. First stretch the patch while keeping the adjacent fluid unchanged. Its material points remain in the same interfacial state, but their deformation and current area change. The mechanical surface stress measures the work of this operation. Now hold the patch at a prescribed deformation and change the adjacent phase from gas to liquid. This changes its wetting state and compares the two surface potentials at that deformation. The energy difference governing this conversion need not equal the difference between the stresses required to stretch the two states. Both quantities derive from the same constitutive potentials, but they describe different variations of those potentials \cite{HuiJagota2013,Makkonen2014}.

The area measure matters for the same reason. A reference-area potential tracks energy for a fixed amount of surface material, whereas a current-area potential measures energy per unit of the deformed area. Their conversion uses the local area ratio. In general, mechanical stress is obtained by differentiating the specified potential with respect to deformation; it cannot be read directly from either energy density. For a simple fluid interface with strain-independent equilibrium potential, the same scalar both assigns the cost of area and supplies isotropic surface stress. A material solid surface retains a strain reference and may store stretch and shear energy. Even an initially isotropic surface can respond anisotropically after prestrain \cite{HeydenEtAl2021}.

The surface constitutive law must therefore be specified alongside the bulk law. This independence concerns the stored energy and its material parameters. A coherently bonded surface still moves with the boundary of the bulk solid: its deformation is inherited from that placement, and it does not form an independently slipping layer. The permanent bulk memory provides the material reference beneath the interface; it does not determine how wetting or stretching that interface changes its free energy. The finite-strain derivative, with its reference and current measures, is given in Appendix~\ref{sec:surface-stress-exact}.

The equilibrium surface potentials also depend on what is held fixed during conversion. The thermodynamic ensemble specifies these imposed controls and conservation constraints. If adsorbed species exchange with reservoirs, the relevant potential includes the chemical work of that exchange; this is the surface grand potential. If total adsorbed amounts are conserved instead, their conservation remains among the equilibrium constraints. Stretching a surface at fixed total adsorbed amount changes its concentration, so its Helmholtz energy per current area is generally different from the potential conjugate to area change. The underlying solid material itself is not an exchangeable adsorbate. The mathematical reduction is given in Appendix~\ref{sec:technical-details}.

These controls determine the wet and dry equilibrium states whose potentials enter Young's law. Temperature, adsorption, pressure and prescribed surface strain can therefore change the equilibrium angle on a chemically homogeneous wall \cite{Rey2004,WardWu2007,GhasemiWard2010,WardSefiane2010}. Such a shift changes the equilibrium target. It must be distinguished from a moving line that has not reached that target, or a line trapped by a metastable barrier at unchanged controls. The present reversible calculation presumes equilibrium surface states and unrestricted migration; it supplies neither adsorption relaxation kinetics nor a law for contact-line drag or pinning.

On a rigid wall of fixed total area, reversible conversion depends only on the wet--dry energy difference. Adding a common constant to both potentials leaves that specific variation unchanged. Deformation or changes of surface state still require the full constitutive potentials and their derivatives.

An alternative description treats an interface as an isotropic surface phase and expresses its stress through a scalar surface pressure. This is an interfacial stress with units of force per length, distinct from bulk pressure. In the convention used by interface-formation theory, the equilibrium surface pressure is the negative of the tension. That theory also assigns state variables, such as surface density, and relaxation dynamics to the interfaces \cite{Shikhmurzaev2007}. The pressure language therefore belongs to a specified constitutive model. Within this fluid-like equilibrium description, rewriting tensions as negative surface pressures leaves Young's relation unchanged.

Two distinctions are necessary. Adsorption thermodynamics also uses ``surface pressure'' to mean a reduction from a specified reference tension; that convention differs from the negative-tension convention. Moreover, a scalar surface pressure cannot represent general anisotropic solid surface stress or its strain dependence. The present material-surface formulation consequently retains separate equilibrium potentials and mechanical surface stresses, while including a liquid-like isotropic response as a special constitutive choice.

In what follows the surface geometry remains bonded to the solid, while its wet and dry constitutive states may differ. We retain local first-gradient surface energies, which depend on local surface deformation but not on curvature or spatial gradients of that deformation. Bending, surface-gradient and independent line energies are omitted.

\section{Two variations on a deformable support}
\label{sec:deformable}

Partition the reference surface into wet and dry regions $\mathcal S_{SL}$ and $\mathcal S_{SG}$, separated by a material-coordinate boundary $\mathcal C$. The current line is $\Gamma=\vct\chi_S(\mathcal C)$. Let $\vct X$ denote reference position. With bulk deformation gradient $\vct F_S=\nabla_{\!\vct X}\vct\chi_S$ and stored energy density $W_S$ per reference volume, a representative constrained equilibrium energy is
\begin{equation}
\begin{aligned}
\mathcal E={}&\int_{\mathcal B_S}W_S(\vct F_S)\dd V_0
+\int_{\mathcal S_{SL}}\Psi_{SL}\dd A_0
+\int_{\mathcal S_{SG}}\Psi_{SG}\dd A_0
+\gamma_{LG}|\Sigma_{LG}|.
\end{aligned}
\label{eq:deformable-energy}
\end{equation}
The measures $dV_0$ and $dA_0$ refer to reference volume and area; $|\Sigma_{LG}|$ is the current liquid--gas area. Solid incompressibility and fixed liquid volume are imposed as constraints. A nontrivial remote boundary portion, disjoint from the migrating region, has prescribed placement independent of $\mathcal C$. This anchoring removes rigid translations and rotations. Its reaction can transmit force while doing zero virtual work. Different loading ensembles require their own external potential.

\subsection{Material displacement and wet--dry migration}

First hold $\mathcal C$ fixed and vary $\vct\chi_S$. This ordinary spatial variation gives bulk elastic equilibrium, surface traction conditions and the mechanical line balance. The solid surface contributions are generated by $\vct\Upsilon$, with any surviving bulk resultant retained. A contour form and its singularity assumptions are given in Appendix~\ref{sec:technical-details}.

Second change the material partition. On a smooth surface with the same boundary value of the deformation gradient on the wet and dry sides, the placement can be held fixed while a strip changes from $SG$ to $SL$. If $\eta$ is the current co-normal advance, positive toward increasing wet area, define its conjugate force by
\begin{equation}
\delta_a\mathcal E=-\int_\Gamma f_\Gamma^{\Conf}\eta\dd\ell.
\label{eq:config-force}
\end{equation}
Here $\delta_a$ denotes variation of the partition, $d\ell$ is current line length, and the superscript $\Conf$ identifies the configurational force, meaning the force conjugate to partition migration. The same energy therefore produces two equilibrium conditions, associated with two different variations.

The distinction between a material point and the geometric contact line is essential here. A marked solid point has a persistent identity and follows the solid placement. The wet--dry boundary identifies which marked points meet liquid and which meet gas. As the boundary advances, successive points change their adjacent phase, while the points themselves may remain fixed in space. A boundary described in material coordinates is therefore not necessarily attached permanently to the same material points. Nor must the geometric intersection be material relative to every adjoining constituent. For the solid, reversible wetting is migration of a boundary between interfacial states; it does not require transport of solid mass along the wall. This description neither assigns stored mass to the contact line nor presumes a first-order phase transition within the two-dimensional surface.

In an actual sequence of equilibria, migration generally changes the elastic deformation and the liquid shape together. Separating their virtual variations identifies which equilibrium conditions belong to material displacement and which belong to changes of the partition. If displacement, interface shape and internal variables are equilibrated at each partition position, eliminating them gives a reduced energy whose derivative is the migration condition.

That derivative must respect the constraints used in elimination. At equilibrium, contributions arising solely from changes of the eliminated fields cancel in the full constrained first variation. The bulk elastic energy may change along the equilibrium sequence, but its induced change must not be added again to a migration derivative that already includes re-equilibration. Explicit contributions can remain when migration changes a constitutive assignment, transports an elastic defect (a localised incompatibility of the reference deformation) or alters a constraint \cite{Olives2010,HuiJagota2015,LiuEtAl2020Membrane,SnoeijerRolleyAndreotti2018}. Relabeling wet and dry patches on an otherwise smooth homogeneous bulk does not itself create an additional bulk energy-release force. The existence of a ridge alone likewise establishes none of these additional terms.

The localised mechanical force also depends on the chosen description. Ideal dividing-surface, resolved continuum and molecular models need not assign the same separate contributions to bulk and surface forces, even when their complete force balances are consistent \cite{Olives2010,Pandey2020,LiangEtAl2018}. Appendix~\ref{sec:technical-details} gives the constrained envelope identity and the estimates needed to distinguish mechanical and configurational contributions from singular bulk fields. At a ridge, the one-sided attachment conditions determine which simultaneous endpoint and material variations are admissible, as considered next.

\subsection{Compatibility at a ridge}
\label{sec:corner-benchmark}

Work in a plane section per unit unchanged line length. Let $X$ be reference surface arclength, with wet material $X<a_0$ and dry material $X>a_0$. The one-sided surface placements $\vct r_{S\alpha}(X)$ meet at the current endpoint $\vct r_\Gamma$, and
\begin{equation}
\partial_X\vct r_{S\alpha}=\lambda_{S\alpha}\vct t_{S\alpha},
\qquad
\vct N_{SL}=-\vct t_{SL},\quad \vct N_{SG}=\vct t_{SG},
\label{eq:corner-orientation}
\end{equation}
where $\lambda_{S\alpha}>0$ is surface stretch and both unit tangents $\vct t_{S\alpha}$ point toward increasing $X$. Differentiating the attachment condition gives
\begin{equation}
\delta\vct r_{S\alpha}(a_0)
=\delta\vct r_\Gamma
-\lambda_{S\alpha}\vct t_{S\alpha}\delta a_0.
\label{eq:corner-compatibility}
\end{equation}
Thus independently changing the spatial endpoint and its material label requires compatible one-sided placement variations \cite{Olives2010}.

For surface energy $\psi_i(\lambda_i)=\lambda_i\gamma_i(\lambda_i)$ per reference area, with $i=SL,SG$ and primes denoting stretch derivatives, define the tensile stress $\Upsilon_i$ and material surface potential $H_i$ by
\begin{equation}
\Upsilon_i=\psi_i',\qquad
H_i=\psi_i-\lambda_i\psi_i'=-\lambda_i^2\gamma_i'.
\label{eq:corner-conjugates}
\end{equation}
Integration by parts and \eqref{eq:corner-compatibility} give the surface boundary variation
\begin{equation}
\begin{aligned}
\delta\mathcal E_\Sigma\big|_\Gamma
={}&-\left(\Upsilon_{SL}\vct N_{SL}
+\Upsilon_{SG}\vct N_{SG}+\gamma_{LG}\vct N_{LG}\right)
\cdot\delta\vct r_\Gamma\\
&+(H_{SL}-H_{SG})\delta a_0.
\end{aligned}
\label{eq:corner-first-variation}
\end{equation}
If the remaining line-localized bulk/loading variation is written
$-\vct R_b\cdot\delta\vct r_\Gamma-\mathcal G_b\delta a_0$, total stationarity requires
\begin{align}
\Upsilon_{SL}\vct N_{SL}+\Upsilon_{SG}\vct N_{SG}
+\gamma_{LG}\vct N_{LG}+\vct R_b&=\vct0,
\label{eq:corner-mechanical}\\
H_{SL}-H_{SG}-\mathcal G_b&=0.
\label{eq:corner-material}
\end{align}
The vector $\vct R_b$ is the bulk/loading mechanical resultant and $\mathcal G_b$ the configurational driving force toward increasing $a_0$, with signs fixed by their displayed work pairing. These are the mechanical and material conditions in the same coordinates. The resultants must be evaluated from the selected bulk problem. For a homogeneous unpinned ridge with sufficiently weak singularities, $\mathcal G_b=0$ and
\begin{equation}
\lambda_{SL}^2\gamma_{SL}'=\lambda_{SG}^2\gamma_{SG}'.
\label{eq:corner-nopinning}
\end{equation}
The quantities $H_i$ have units of force per length and concern conversion of surface material; they are distinct from chemical potentials for exchanging adsorbate. This is the established surface material-potential condition \cite{SnoeijerRolleyAndreotti2018,Pandey2020}. It does not require equal one-sided stretches. The smooth-surface measure conversion, the corner boundary calculation and the conditions for vanishing bulk resultants are detailed in Appendix~\ref{sec:technical-details}.

\section{Stiffening the supported solid}
\label{sec:stiff}

The two limits serve different purposes. First take $T_*/\lambda\to0$ at fixed $G_S>0$, where $T_*$ is a fixed observation time, to establish permanent bulk memory. Then increase $G_S$ at fixed anchoring and fixed surface constitutive data. Choose a fixed macroscopic reference length $L$ of the boundary-value problem, such as a reference drop radius or a support dimension, independent of $G_S$. The dimensionless elastocapillary ratio is
\begin{equation}
\epsilon_{\ec}=\frac{\Gamma_*}{G_SL}\to0.
\label{eq:two-small-parameters}
\end{equation}
Here $\Gamma_*$ is a fixed surface-stress scale, with units of force per length, drawn from the liquid tension, solid surface stresses evaluated at the prescribed outer-wall state and fixed surface elastic coefficients. It is not an assumed uniform supremum of the stress at a ridge tip. The corresponding deformation length is $\ell_{\ec}=\Gamma_*/G_S$. In this reduction, \emph{the outer scale} means distances comparable with the chosen $L$ and large compared with $\ell_{\ec}$; the near-line region has distances comparable with $\ell_{\ec}$. Different support models can have additional lengths, as discussed in Section~\ref{sec:landscape}.

Rigidity concerns the response of the equilibrium family, rather than the absence of motion in any one equilibrium. A compliant solid at rest already has zero rate of deformation; it can nevertheless sustain a finite, statically deformed wetting ridge. Incompressibility likewise preserves volume while allowing shear and changes of shape. The required rigid-support passage must therefore control the capillary-induced displacement and strain relative to the prescribed placement. A condition on velocity or strain rate alone cannot establish this passage. The distinction is particularly important after the no-relaxation limit, since the stored deformation remains present even when all velocities vanish.

The small elastocapillary ratio also admits physically different interpretations. Increasing the solid modulus while keeping the body, capillary data and anchoring fixed changes the mechanical response of a given boundary-value problem. Observing a fixed compliant support over distances much larger than its elastocapillary length can instead make its localised deformation small relative to the observation length. The second procedure need not reduce the dimensional ridge displacement or alter its local surface stretches. Changing the drop size or support dimensions introduces further changes to the loading and geometry. The stiffening construction considered here keeps these choices fixed; a similar macroscopic appearance does not by itself identify the same limiting problem.

At the outer scale just defined, the usual capillary scaling gives displacement divided by $L$, and strain increments, of order $\epsilon_{\ec}$, up to geometry-dependent logarithmic factors. Their product with $G_S$ need not vanish. The outer wall can therefore become undeformed while transmitting a finite load. For a prescribed prestretch, these are increments about that placement; its baseline stress may itself grow with $G_S$.

The spatial resolution of the contact angle must be specified for the same reason. At a fixed distance from the line, increasing stiffness moves the observation point progressively farther from the ridge when distance is measured in elastocapillary lengths. The liquid interface can then approach the Young angle against the prescribed wall even if the ridge retains nontrivial tangents when examined within its shrinking neighbourhood. An angle measured against either ridge flank is a different quantity. Approaching the line first retains the local geometry; stiffening first at a fixed observation distance gives the wall geometry. These operations need not commute, and convergence of the macroscopic angle imposes no general requirement that the one-sided ridge angles coincide with it \cite{StyleDufresne2012,StyleEtAl2013,LubbersEtAl2014}.

The limiting solid mechanics is represented by a support reaction. It includes distributed tractions and any concentrated contribution generated by the limit. A remote clamp need not act as a literal point force at the line. An unanchored body, by contrast, retains rigid translation and rotation and does not automatically produce a fixed-wall problem.

Anchoring specifies more than the removal of an arbitrary rigid translation. It also determines which external work is included when the wet--dry partition moves. Prescribed displacement and prescribed applied force are different loading controls: in the latter case the work of the applied force belongs to the potential being varied. If a homogeneous prestretch is held fixed, migration takes place at that prescribed stretch, and the associated equilibrium surface energies are evaluated there. Allowing the controls themselves to change with migration would define a different variational problem. Consequently, the support constraint and its external work must be specified before identifying the surviving contact-line variation.

Holding a nonzero prestretch fixed while increasing the modulus can require an increasing baseline stress. The bounded capillary load then produces a small increment about an already stressed state; it does not make the total stress small. Prescribing a fixed external stress instead generally permits the baseline stretch to change with stiffness, and hence can change the surface states at which Young's law is evaluated. Similarly, allowing surface elastic moduli or surface prestresses to increase with the bulk modulus can preserve additional deformation lengths or alter the capillary scaling. These are physically distinct families, excluded by the fixed controls in the present passage.

\subsection{The surviving variation}

On a prescribed flat placement, let $\lambda_\infty$ be the common imposed surface stretch. Compatibility reduces to
\begin{equation}
\delta\vct r_\Gamma
=\delta\vct\chi_S(a_0)
+\lambda_\infty\vct t_W\delta a_0,
\qquad \delta\vct\chi_S(a_0)=\vct0,
\label{eq:rigid-compatibility}
\end{equation}
where $\vct t_W$ is the unit wall tangent normal to the contact line and pointing toward increasing wet area. The endpoint follows the migrating material label even though the solid material is fixed.

The surface boundary form shows exactly which condition survives. Let $\theta$ be the liquid-side angle of the interface against this flat wall. On the common flat trace, define the tangential sum of surface-stress resultants by
$S_t=\Upsilon_{SG}-\Upsilon_{SL}-\gamma_{LG}\cos\theta$.
Using $H_i=\lambda_\infty(\gamma_i-\Upsilon_i)$ in \eqref{eq:corner-first-variation} gives
\begin{equation}
\begin{aligned}
\delta\mathcal E_\Sigma\big|_\Gamma
&=\left[-\lambda_\infty S_t+H_{SL}-H_{SG}\right]\delta a_0\\
&=\lambda_\infty
\left(\gamma_{SL}-\gamma_{SG}+\gamma_{LG}\cos\theta\right)\delta a_0.
\end{aligned}
\label{eq:restricted-corner-young}
\end{equation}
The stresses cancel in the restricted variation. The two independent free-corner equations are therefore not imposed separately on the rigid wall. A ridge resolved at distances comparable with $\ell_{\ec}$ may still satisfy them with different one-sided states.

A mechanical reaction $\vct R$ acting on solid material has the corresponding virtual work $\delta W_R$:
\begin{equation}
\delta W_R=\vct R\cdot\delta\vct\chi_S(a_0)
=\vct R\cdot
\left(\delta\vct r_\Gamma-\lambda_\infty\vct t_W\delta a_0\right)=0.
\label{eq:reaction-material-pairing}
\end{equation}
A tangential reaction is consequently compatible with reversible migration. Assigning it work solely through the motion of the geometric endpoint would use a displacement that is not work-conjugate to the reaction.

\subsection{What must converge}
\label{sec:first-variation-limit}

The restriction above is exact for the prescribed target geometry. Passing from compliant equilibria to that target requires more than vanishing displacement. For a differentiable stationary branch with migration coordinate $a$, let $\mathcal E_{\rm red}^{\epsilon_{\ec}}$ be the energy after the remaining fields have equilibrated, and let $\mathcal E_W$ be the prescribed-wall energy. A sufficient hypothesis is
\begin{equation}
\mathcal E_{\rm red}^{\epsilon_{\ec}}(a)
=\mathcal E_W(a)+C_{\epsilon_{\ec}}+r_{\epsilon_{\ec}}(a),
\qquad \partial_a r_{\epsilon_{\ec}}\to0
\quad\text{locally uniformly},
\label{eq:first-variation-convergence}
\end{equation}
where $C_{\epsilon_{\ec}}$ is independent of migration and $r_{\epsilon_{\ec}}$ is the residual energy. It excludes a residual configurational contribution from the eliminated mechanics. A small residual energy can still vary rapidly with the partition position and produce a finite migration force. Vanishing displacement or energy therefore cannot replace control of the migration variation; a simple counterexample is given in Appendix~\ref{sec:technical-details}.

Global almost minimisers admit another route. Uniform convergence of reduced energies on a compact parameter interval passes minimality to a continuous target energy, even without derivative convergence. Appendix~\ref{sec:nonlinear-continuum} establishes this route for a finite-strain cap class; Appendix~\ref{sec:quadratic-continuum} establishes the stationary-derivative route in its defined quadratic model. The general nonlinear stationary passage remains conditional. The constrained envelope identity and a weak residual formulation are given in Appendix~\ref{sec:technical-details}.

\section{Young equilibrium and the retained reaction}
\label{sec:young}

\subsection{Area conversion on the fixed wall}

Let $\theta$ denote the liquid-side angle of a candidate configuration. For the fixed homogeneous wall, evaluate $\gamma_{SL}$ and $\gamma_{SG}$ in their specified equilibrium states at the prescribed placement. A current co-normal line advance $\delta a>0$ along $\vct t_W$ converts dry to wet area:
$\delta A_{SL}=\delta a\,d\ell$ and $\delta A_{SG}=-\delta a\,d\ell$.
The liquid--gas endpoint contributes
\begin{equation}
\delta A_{LG}\big|_\Gamma
=-\vct N_{LG}\cdot\vct t_W\,\delta a\,d\ell
=\cos\theta\,\delta a\,d\ell.
\label{eq:LG-endpoint-variation}
\end{equation}
After the liquid interface satisfies its volume-constrained interior equilibrium, the remaining variation is
\begin{equation}
\delta\mathcal E\big|_\Gamma
=\int_\Gamma
\left(\gamma_{SL}-\gamma_{SG}+\gamma_{LG}\cos\theta\right)
\delta a\,d\ell.
\label{eq:young-energy-variation}
\end{equation}
Arbitrary reversible migration therefore gives
\begin{equation}
f_\Gamma^{\Conf}
=\gamma_{SG}-\gamma_{SL}-\gamma_{LG}\cos\theta=0,
\qquad
\gamma_{SG}-\gamma_{SL}=\gamma_{LG}\cos\theta_Y.
\label{eq:young-derived}
\end{equation}
An interior partial-wetting angle requires $|\gamma_{SG}-\gamma_{SL}|<\gamma_{LG}$. Equality is degenerate. If $\gamma_{SG}-\gamma_{SL}>\gamma_{LG}$, the sharp area model favors spreading; if $\gamma_{SG}-\gamma_{SL}<-\gamma_{LG}$, it favors loss of wetted area. Neither case selects an interior finite-angle line, and resolving films or wetting transitions requires additional physics. Young's equation is a stationarity condition, with stability and global selection requiring their own argument.

The liquid-volume constraint has already supplied the pressure jump in the interior interface equation, the Young--Laplace relation between pressure and curvature. Its multiplier does not replace the endpoint condition. Interior shape equilibrium and Young's boundary condition are the two corresponding parts of one constrained variation; the full calculation is given in Appendix~\ref{sec:technical-details}.

The scalar character of Young's equation follows from the remaining motion. Locally, a contact line constrained to the wall advances in the wall-tangent direction perpendicular to itself. Translation along the line changes its parametrization rather than converting wet and dry area. Displacement normal to the wall belongs to the constrained solid-placement problem and has its own mechanical reaction. The migration condition therefore has one scalar coefficient, even though the complete mechanical balance remains vectorial. The absence of a normal capillary term from Young's equation reflects this admissible variation; it does not remove the normal load from the physical system.

\subsection{Mechanical closure}

The scalar condition does not exhaust the force balance. With $\vct n_W$ the unit wall normal pointing out of the solid,
\begin{equation}
\gamma_{LG}\vct N_{LG}
=-\gamma_{LG}\cos\theta\,\vct t_W
+\gamma_{LG}\sin\theta\,\vct n_W.
\label{eq:capillary-decomp}
\end{equation}
The normal edge load remains mechanically present. Its transmission may be concentrated or distributed, depending on the model \cite{White2003,MarchandSoft2012}. In the outer description,
\begin{equation}
R_n+\gamma_{LG}\sin\theta
+\text{other resolved normal contributions}=0.
\label{eq:normal-reaction}
\end{equation}
Here $R_n$ is the normal component of the effective solid/support reaction. This is a local resultant statement.

The force transmitted near the contact line must be distinguished from the total force on a remote clamp. A small region surrounding the line detects the liquid--gas edge load and the stresses that transmit it through the solid. A region containing the entire wetted footprint also includes the pressure exerted by the liquid. For a weightless equilibrium drop these contributions can have zero total normal resultant, while remaining nonzero locally and producing a ridge in a compliant support. Thus a vanishing net clamp force would not establish the absence of capillary loading. Conversely, the normal line load alone cannot determine the clamp force without the remaining distributed loads and boundary conditions.

Let $\Upsilon_{S\alpha}^{tt}=\vct t_W\cdot\vct\Upsilon_{S\alpha}\vct t_W$ be the surface-stress component along the advancing wall tangent, and let $R_t^{\rm eff}$ denote the corresponding effective solid/support reaction. With no additional tangential force, the balance is
\begin{equation}
R_t^{\rm eff}+\Upsilon_{SG}^{tt}-\Upsilon_{SL}^{tt}
-\gamma_{LG}\cos\theta=0.
\label{eq:rigid-tangential-mechanics}
\end{equation}
At Young equilibrium this requires
\begin{equation}
R_t^{\rm eff}
=(\gamma_{SG}-\Upsilon_{SG}^{tt})
-(\gamma_{SL}-\Upsilon_{SL}^{tt}).
\label{eq:tangential-reaction-shuttleworth}
\end{equation}
This effective resultant accounts for the eliminated solid mechanics; a microscopic line localisation needs a separate argument. It vanishes in this balance when both surfaces are liquid-like, $\vct\Upsilon_{S\alpha}=\gamma_{S\alpha}\vct P_W$, where $\vct P_W$ projects onto the wall tangent plane. Otherwise the reaction reconciles the mechanical stresses with the energy-based migration condition. Its zero-work pairing is \eqref{eq:reaction-material-pairing}.

A tangential reaction also has a concrete mechanical meaning without constituting resistance to wetting. It balances the tangential resultant generated by the liquid tension and the solid surface stresses, and is transmitted by the constrained solid. Friction or pinning concerns a different operation: resistance to migration of the wet--dry partition through the surface material. The reaction in the present model is paired with displacement of that material, which remains fixed. Its nonzero value therefore neither supplies a dissipation rate nor establishes a threshold for advancing or receding motion. Such a threshold requires an additional constitutive mechanism or a restriction on the admissible migration.

Mechanically closed control-volume and molecular arguments can also recover Young's relation \cite{YamaguchiEtAl2019,FernandezToledanoEtAl2017}. The distinction here concerns the work-conjugate variables: Young's equation tests wet--dry conversion on the fixed wall, while the complete vector balance includes the solid surface stresses and reaction.

\subsection{A surface law that makes the difference explicit}
\label{sec:benchmark}
\label{sec:quadratic-benchmark}
\label{sec:benchmark-rigid}

Consider the reference-area law in plane strain, meaning no displacement or dependence in the transverse direction,
\begin{equation}
\psi_i(\lambda_s)
=\gamma_{0,i}\lambda_s+b_i(\lambda_s-1)
+\frac{k_i}{2}(\lambda_s-1)^2,
\qquad i=SL,SG,
\label{eq:benchmark-law}
\end{equation}
where $\lambda_s$ is scalar surface stretch, on an admissible interval with positive surface stress and $k_i\ge0$. Here $\gamma_{0,i}=\psi_i(1)$ is the reference-state surface energy, $b_i$ the surface-stress offset from that energy, and $k_i$ the surface elastic modulus. The offset $b_i$ specifies the surface response, not bulk prestress. All have units of force per length. Its current-area energy is $\psi_i/\lambda_s$, and its stress is $\psi_i'$; the full derivatives and corner condition are given in Appendix~\ref{sec:technical-details}.

Positive surface stress stabilizes transverse distortions of a surface with no bending resistance; a nonnegative stretching modulus alone does not exclude instability under compression. The coefficients allow the wet and dry energies, stresses and stretching responses to differ. At a freely migrating compliant corner, their material-potential condition can require different one-sided stretches, even when the two surfaces share a common distant imposed stretch. That local condition must be solved together with the bulk equilibrium; a prescribed tip stretch by itself is not a ridge solution.

For an illustrative uncalibrated choice, set $\lambda_\infty=5/4$, $b_i=0$, $\gamma_{0,SL}=40$, $\gamma_{0,SG}=50$, $k_{SL}=20$, $k_{SG}=40$, and $\gamma_{LG}=30$, with coefficients in $\mathrm{mN\,m^{-1}}$. Then
\begin{equation}
\begin{gathered}
\gamma_{SL}=40.5,\quad\gamma_{SG}=51,
\qquad \Upsilon_{SL}=45,\quad\Upsilon_{SG}=60,\\
\theta_Y\simeq69.51^\circ,
\qquad R_t^{\rm eff}=-4.5\ \mathrm{mN\,m^{-1}}.
\end{gathered}
\label{eq:benchmark-numbers}
\end{equation}
The configurational balance uses the energy difference $10.5$; the mechanical balance uses the stress difference $15$ and the reaction $-4.5$. Setting that reaction to zero would select a different angle. These are values on the prescribed common flat placement, not an identification of the one-sided stretches of a compliant ridge.

Nor is macroscopic prestretch necessary for a tangential reaction. At unit imposed stretch, the offsets $b_i$ can already make the stress difference unequal to the energy difference. The reaction then supplies their mismatch while the reversible migration condition remains Young's law.

\begin{figure}[tbp]
\centering
\resizebox{0.98\textwidth}{!}{%
\begin{tikzpicture}[x=1cm,y=1cm,>=Latex,font=\small]
\begin{scope}[xshift=-4.25cm]
\node[font=\bfseries] at (0,3.00) {(a) mechanical balance};
\coordinate (O) at (0,0);
\draw[very thick] (-3.40,0)--(3.40,0);
\draw[very thick] (O)--(-2.75,2.05);
\fill (O) circle (1.25pt);
\node[below=3.0mm] at (-2.25,0.15) {$SL$};
\node[below=3.0mm] at (2.25,0.15) {$SG$};
\node[anchor=south east] at (-2.64,2.18) {$\Sigma_{LG}$};
\node[above right=0.8mm and 0.8mm] at (O) {$\Gamma$};

\draw[line width=0.65pt] (-0.94,0)
  arc[start angle=180,end angle=143.3,radius=0.94];
\node at (-1.16,0.50) {$\theta$};

\draw[->,very thick] (O)--(-2.17,1.62);
\node[anchor=south west] at (-1.72,1.46) {$\gamma_{LG}\vct N_{LG}$};
\node[draw,rounded corners,align=left,anchor=west,inner sep=2.0mm,
      text width=2.75cm,font=\tiny]
  at (0.55,1.46)
  {wall-basis components:\\[0.7mm]
   $\vct t_W:\;-\gamma_{LG}\cos\theta$\\[0.7mm]
   $\vct n_W:\;+\gamma_{LG}\sin\theta$\\[0.9mm]
   normal balance:\\[-0.2mm]
   $R_n+\gamma_{LG}\sin\theta+\cdots=0$};
\end{scope}

\draw[densely dotted] (0,-1.70)--(0,3.08);

\begin{scope}[xshift=4.25cm]
\node[font=\bfseries] at (0,3.00) {(b) configurational variation};
\coordinate (O) at (0,0);
\draw[very thick] (-3.40,0)--(3.40,0);
\draw[very thick] (O)--(-2.75,2.05);
\fill (O) circle (1.25pt);
\node[below=3.0mm] at (-2.25,0.15) {$SL$};
\node[below=3.0mm] at (2.25,0.15) {$SG$};
\node[anchor=south east] at (-2.64,2.18) {$\Sigma_{LG}$};
\node[above right=0.8mm and 0.8mm] at (O) {$\Gamma$};
\draw[line width=0.65pt] (-0.94,0)
  arc[start angle=180,end angle=143.3,radius=0.94];
\node at (-1.16,0.50) {$\theta$};

\draw[->,very thick] (0.50,-1.13)--(2.88,-1.13)
  node[midway,above=0.0mm] {$\delta a>0$};
\node[align=center,font=\scriptsize] at (1.69,-1.68)
{$\Gamma$ advances; a strip of $SG$\\is converted into $SL$};

\node[draw,rounded corners,align=left,anchor=west,inner sep=2.3mm,font=\scriptsize]
  at (0.82,1.48)
  {$\delta A_{SL}=+\delta a\,d\ell$\\
   $\delta A_{SG}=-\delta a\,d\ell$\\
   $\delta A_{LG}|_{\Gamma}=+\cos\theta\,\delta a\,d\ell$};
\end{scope}
\end{tikzpicture}%
}
\caption{Mechanical loading and configurational migration in the same local geometry. (a) The liquid--gas edge traction has tangential and normal components, both retained in the complete mechanical balance. (b) At fixed solid placement, advance of the wet--dry partition converts solid--gas area into solid--liquid area. The liquid--gas endpoint variation supplies the cosine term in Young's law. The wall tangent points to the right and the normal points upward; all resultants are per unit current contact-line length.}
\label{fig:young-unravelled}
\end{figure}
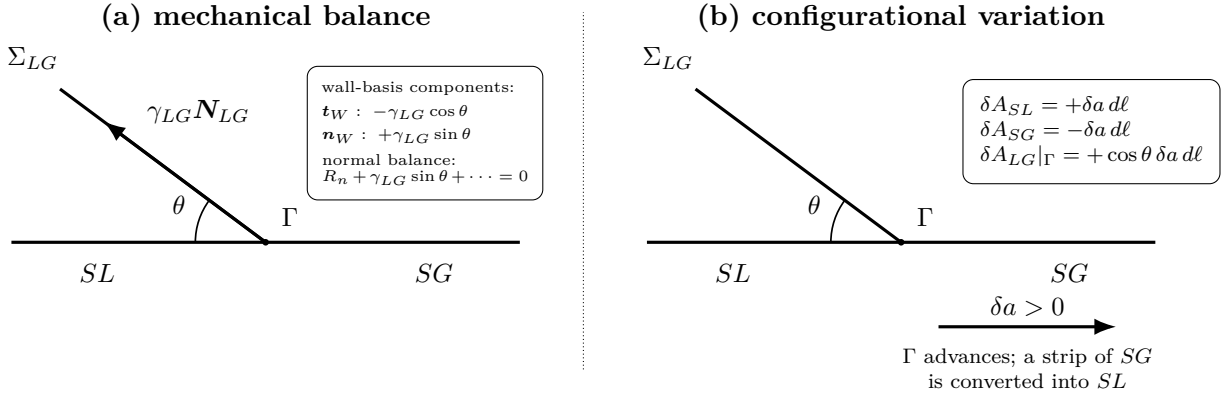

\section{What the continuum results establish}
\label{sec:questions}
\label{sec:scales}

The central limiting question is whether the migrating wet--dry partition retains only the fixed-wall energy variation. The two continuum calculations answer it in different settings.

\begin{table}[tbp]
\centering
\small
\caption{Scope of the supporting results. Each conclusion is for its stated model and admissible class.}
\label{tab:proof-scope}
\vspace{6pt}
\begin{tabular}{@{}>{\raggedright\arraybackslash}p{0.18\textwidth}>{\raggedright\arraybackslash}p{0.38\textwidth}>{\raggedright\arraybackslash}p{0.36\textwidth}@{}}
\toprule
Calculation & Established & Not supplied by that calculation\\
\midrule
Constitutive (App.~\ref{sec:transport-estimate}) & Finite-time conformation and stress estimates along a prescribed common motion & Convergence of the coupled wetting motion\\[1mm]
Finite strain (App.~\ref{sec:nonlinear-continuum}) & Uniform reduced-energy convergence and spatial Young-angle recovery for global almost minimisers & Arbitrary stationary branches or a limiting reaction distribution\\[1mm]
Quadratic (App.~\ref{sec:quadratic-continuum}) & Stationary migration derivatives, angle reconstruction and traction concentration & Uniform approximation of the nonlinear near-line deformation\\
\bottomrule
\end{tabular}
\end{table}

The finite-strain model has a fixed incompressible neo-Hookean body in plane strain, clamped at its undeformed side and bottom boundaries. Its upper surface remains a graph, carries two possibly different tensile surface laws with positive quadratic stretch penalties, and supports a symmetric weightless drop of fixed cross-sectional area with a circular liquid--gas interface. The volume constraint accounts for the area displaced by the solid beneath the cap. Bulk and surface energy bounds, together with this exact liquid geometry, control the moving endpoint. Corollary~\ref{cor:nl-reduced-energy} and Theorem~\ref{thm:nl-young} then pass global minimality to the flat Young state. The proof permits strain concentration near the line. Its quadratic lower bound on surface stretching excludes the liquid-like law $\psi_i=\gamma_{0,i}\lambda_s$, for which a different estimate is needed.

The second model is an incompressible linear-elastic half-space with common positive wet/dry surface stretch moduli and reference surface stresses, although its reference surface energies may differ. The displacement decays at infinity, fixing the remote reference placement. Its reference circular cap encloses a fixed liquid cross-sectional area, and the capillary energy is retained to first order in solid displacement. The surface coefficients regularise displacement near the point loads. Including liquid pressure cancels the total normal line load and prevents an energy divergence from arbitrarily long wavelengths. Minimizing the complete energy lowers it below the flat-placement value; using only the positive stored elastic energy would give the wrong migration correction.

The migration derivative converges faster than the energy correction because of a cancellation specific to this loading. Equal wet/dry reference stresses leave a contact-line force of constant magnitude, the liquid--gas tension, whose direction changes with the cap angle. The normal and tangential responses have the same leading logarithmic dependence on stiffness. Its coefficient in the minimised energy is independent of the force direction, so that stiffness-dependent logarithm is absent from the migration derivative. Unequal reference stresses also change the surface stiffness, so the result is not asserted for that case. Displacement changes both the footprint and the circular-segment area. The reconstructed spatial angle consequently has a logarithmic leading correction absent from the flat-reference angle. This reconstruction uses first-order geometry; it is not a nonlinear equilibrium solution. Both angles refer to the fixed horizontal, not the local ridge tangent.

The same model resolves the retained reaction. At finite stiffness, surface displacement derivatives have finite jumps while bulk traction has only an integrable logarithmic singularity. Stiffening while observing loads on the fixed droplet scale makes displacement vanish but leaves a limiting traction distribution containing the line load. Shrinking the observation neighbourhood about the line first instead gives zero bulk resultant. In the fixed-scale limit, the concentrated surface force associated with a derivative jump is cancelled by the neighbouring distributed surface force; a local jump alone therefore does not identify the transmitted load at the droplet scale. Surface elasticity does not bound every bulk strain: a nonzero tangential point load gives a logarithmic strain singularity near the line. Thus this calculation does not uniformly approximate a finite-strain ridge there, or supply the nonlinear theorem with a reaction-convergence result. The calculations and numerical angle reconstruction are in Appendix~\ref{sec:quadratic-continuum}.

This explains how a Neumann-like ridge resolved near the line and a Young angle measured against the outer wall can coexist \cite{StyleDufresne2012,StyleEtAl2013,LubbersEtAl2014}. They concern different spatial resolutions and different variations. The nonlinear almost-minimiser theorem and the quadratic traction calculation support this interpretation without proving convergence of every nonlinear stationary branch.

The equilibrium construction accordingly does not prescribe a contact-line speed or a dynamic contact angle. In a moving problem, liquid flow, relaxation of an already solid support, and surface-state kinetics can all contribute to dissipation. Their laws must accompany the mechanical and configurational balances. The no-relaxation passage used to construct the elastic bulk should therefore be distinguished from viscoelastic dissipation caused by motion of a ridge through an existing solid \cite{KarpitschkaEtAl2015,DervauxEtAl2020,EssinkEtAl2025}. Young's relation identifies the reversible stationary endpoint under the stated controls; it does not alone establish that a specified dynamics reaches that endpoint.

The restrictions on equilibrium remain physical ones. A retained line energy contributes its own migration variation and can produce curvature-dependent corrections; roughness, disorder and metastable barriers can obstruct reversible migration. Surface-state changes can alter the equilibrium energy difference, while viscoelastic or poroelastic dissipation affects the motion toward equilibrium. These extensions require their corresponding energy, constraints or kinetic laws. Section~\ref{sec:landscape} places them among the neighbouring support models.

\section{Support models, force transmission and spatial resolution}
\label{sec:landscape}

The anchored bulk solid considered here is one of several distinct support models. The reviews \cite{StyleEtAl2017,ChenEtAl2018,AndreottiSnoeijer2020} provide broader context. The following comparison identifies the choices of support mechanics, surface law, force transmission and observational scale that distinguish neighbouring models.

\subsection{Support class and mechanical model}

For a thick elastic substrate, liquid--gas tension sets the capillary loading, while solid surface stresses and surface elastic coefficients introduce mechanical lengths. The elastocapillary length compares surface forces with bulk elastic resistance:
\begin{equation}
\ell_{\ec}\sim \frac{\Gamma_*}{G_S}
\quad\text{or}\quad
\frac{\Gamma_*}{E_S},
\label{eq:elastocapillary-length}
\end{equation}
Here $E_S$ is the substrate Young modulus; $G_S$ is its shear modulus as above. The specified stress scale $\Gamma_*$ is formed from the liquid--gas tension $\gamma_{LG}$, solid surface stresses in the reference or distant held state, and the fixed surface elastic coefficients relevant to the model, all expressed as force per unit length. Geometry, Poisson ratio and the surface law determine the precise lengths and numerical factors. This estimate does not assume a uniform bound on ridge-tip stresses; the finite-strain result in \cref{sec:nonlinear-continuum} allows strain concentration.

Jerison et al. resolved a wetting ridge and demonstrated surface-stress regularisation of the near-line response \cite{Jerison2011}. Style and Dufresne showed that a Young-like macroscopic angle can coexist with a Neumann-like microscopic ridge \cite{StyleDufresne2012}. Here the macroscopic angle is measured against the undeformed support plane at distances from the line large compared with the ridge region; the microscopic geometry concerns the local ridge flanks. Pandey et al. treated large ridge deformations with several surface constitutive laws \cite{Pandey2020}.

Membranes are governed by stretching and tension. Their wetting equilibria require both mechanical and configurational conditions, with elastocapillary ratios involving tension, modulus and thickness \cite{HuiJagota2015,LiuEtAl2020Membrane}. Plates and thin sheets additionally resist bending; their capillary response can involve substantial induced tension and stretching \cite{Fortes1984,RomanBico2010,DavidovitchVella2018}. In a viscoelastic support, ridge motion dissipates energy through constitutive relaxation and produces velocity-dependent apparent angles \cite{KarpitschkaEtAl2015,DervauxEtAl2020,EssinkEtAl2025}. In swollen gels, poroelastic relaxation couples solvent redistribution to deformation of the porous elastic network \cite{ZhengEtAl2025}. These mechanisms require different limiting parameters and are not covered by the fixed-support equilibrium passage.

\subsection{Constitutive description of solid surfaces}

A simple fluid interface $\Sigma$ has reversible surface-stress tensor $\vct S^\Sigma$, equilibrium area potential $\gamma^\Sigma$, and tangent projector $\vct P_\Sigma$, related by
\begin{equation}
\vct S^\Sigma=\gamma^\Sigma\vct P_\Sigma,
\label{eq:fluid-surface-law-landscape}
\end{equation}
so its equilibrium area potential also sets its isotropic mechanical tension. A solid surface potential can additionally depend on strain, orientation, adsorption, composition, order and temperature. In the small-strain Shuttleworth notation, $\Upsilon_{\alpha\beta}$ are the components of the surface-stress tensor $\vct\Upsilon$, $\epsilon^s_{\alpha\beta}$ are surface-strain components, and $\delta_{\alpha\beta}$ is the Kronecker delta in tangent coordinates. With current-area potential $\gamma$,
\begin{equation}
\Upsilon_{\alpha\beta}
=\gamma\delta_{\alpha\beta}
+\frac{\partial\gamma}{\partial\epsilon^s_{\alpha\beta}}.
\label{eq:shuttleworth-landscape}
\end{equation}
The finite-strain form depends on the area measure and stress/strain representation \cite{Shuttleworth1950,GurtinMurdoch1975,AndreottiSnoeijer2016,StyleEtAl2017,HeydenBain2024}. Rey's generalized surface thermodynamics includes simultaneous thermodynamic, strain, order and geometric dependence \cite{Rey2004}. In each case, mechanical stress is obtained from the surface potential by the appropriate constitutive derivative. Energy and mechanics therefore remain connected even when $\gamma$ and $\vct\Upsilon$ differ.

\subsection{Transmission of capillary loading to the support}

Sharp-interface elastic models commonly use a line force together with distributed Laplace pressure. Other models resolve a finite interaction region. White transmitted forces through disjoining pressure, the pressure arising from molecular interactions across a thin film \cite{White2003}; Marchand et al. coupled molecular mean-field forces to elasticity and identified dependence on force transmission within the solid surface layer \cite{MarchandSoft2012}. Bostwick et al. compared line-force models for finite substrates with distinct wet and dry surface stresses \cite{BostwickEtAl2014}. Molecular simulations have retained an explicit elastic contribution in local line-force accounting \cite{LiangEtAl2018}, while nonlinear continuum asymptotics can express the leading local ridge balance through surface edge tractions alone \cite{Pandey2020}.

These descriptions refer to different constitutive and asymptotic settings. The balance in \eqref{eq:mechanical-contour-balance} uses a contour surrounding the contact line and retains both surface edge tractions and the limiting bulk-stress resultant. Whether the bulk contribution vanishes or survives must be established within the selected model.

\subsection{Mechanical and configurational equilibrium conditions}

Material displacement and migration of the wet--dry boundary generate distinct variations of the same energy. This distinction is explicit in deformable-body surface thermodynamics \cite{Olives2010}, membrane theories \cite{HuiJagota2015,LiuEtAl2020Membrane}, and variational treatments of stretched and hyperelastic substrates \cite{SnoeijerRolleyAndreotti2018,Pandey2020}. The present construction applies it to the change in admissible kinematics under a prescribed solid placement. In \cref{tab:model-landscape}, $\Upsilon_S$ denotes a scalar isotropic solid surface stress.

\begin{table}[tbp]
\centering
\caption{Selective comparison of wetting models by support mechanics, surface law and force transmission.}
\label{tab:model-landscape}
\vspace{6pt}
\scriptsize
\renewcommand{\arraystretch}{1.08}
\setlength{\tabcolsep}{2pt}
\begin{tabular}{@{}lllll@{}}
\toprule
\parbox[t]{0.14\textwidth}{\raggedright Support class} & \parbox[t]{0.17\textwidth}{\raggedright Bulk mechanics} & \parbox[t]{0.18\textwidth}{\raggedright Surface constitutive level} & \parbox[t]{0.21\textwidth}{\raggedright Contact-line / force-transmission structure} & \parbox[t]{0.23\textwidth}{\raggedright Representative role}\\
\midrule
\parbox[t]{0.14\textwidth}{\raggedright Elastic half-space/layer} & \parbox[t]{0.17\textwidth}{\raggedright linear elasticity} & \parbox[t]{0.18\textwidth}{\raggedright often constant isotropic $\Upsilon_S$} & \parbox[t]{0.21\textwidth}{\raggedright line/resultant capillarity + pressure; surface stress regularises ridge} & \parbox[t]{0.23\textwidth}{\raggedright ridge mechanics and elastocapillary scale \cite{Jerison2011,StyleDufresne2012,StyleEtAl2013}}\\
\parbox[t]{0.14\textwidth}{\raggedright Hyperelastic solid} & \parbox[t]{0.17\textwidth}{\raggedright finite elasticity} & \parbox[t]{0.18\textwidth}{\raggedright general strain-dependent / Shuttleworth law} & \parbox[t]{0.21\textwidth}{\raggedright variational or local singular ridge balance} & \parbox[t]{0.23\textwidth}{\raggedright large strain, configurational structure \cite{Pandey2020,SnoeijerRolleyAndreotti2018,HenkelEtAl2022}}\\
\parbox[t]{0.14\textwidth}{\raggedright Elastic membrane} & \parbox[t]{0.17\textwidth}{\raggedright finite stretching/tension} & \parbox[t]{0.18\textwidth}{\raggedright wet/dry membrane surface laws} & \parbox[t]{0.21\textwidth}{\raggedright simultaneous mechanical and configurational balances} & \parbox[t]{0.23\textwidth}{\raggedright direct force/configurational decomposition \cite{Fortes1984,HuiJagota2015,LiuEtAl2020Membrane}}\\
\parbox[t]{0.14\textwidth}{\raggedright Thin sheet/plate} & \parbox[t]{0.17\textwidth}{\raggedright stretching + bending} & \parbox[t]{0.18\textwidth}{\raggedright surface energy + sheet law} & \parbox[t]{0.21\textwidth}{\raggedright capillarity coupled to induced tension/bending} & \parbox[t]{0.23\textwidth}{\raggedright bending and induced tension \cite{RomanBico2010,DavidovitchVella2018}}\\
\parbox[t]{0.14\textwidth}{\raggedright Visco-/poroelastic solid} & \parbox[t]{0.17\textwidth}{\raggedright memory and/or solvent transport} & \parbox[t]{0.18\textwidth}{\raggedright often elastic surface law; extensions possible} & \parbox[t]{0.21\textwidth}{\raggedright moving ridge dissipates and may redistribute solvent} & \parbox[t]{0.23\textwidth}{\raggedright time-dependent soft-support response \cite{KarpitschkaEtAl2015,DervauxEtAl2020,ZhengEtAl2025}}\\
\parbox[t]{0.14\textwidth}{\raggedright Distributed interaction} & \parbox[t]{0.17\textwidth}{\raggedright elastic continuum + interaction model} & \parbox[t]{0.18\textwidth}{\raggedright free energies linked to interaction potential} & \parbox[t]{0.21\textwidth}{\raggedright capillary traction spread over interaction zone} & \parbox[t]{0.23\textwidth}{\raggedright probes force-transmission/model dependence \cite{White2003,MarchandSoft2012}}\\
\parbox[t]{0.14\textwidth}{\raggedright Rigid support} & \parbox[t]{0.17\textwidth}{\raggedright placement constrained} & \parbox[t]{0.18\textwidth}{\raggedright specified equilibrium $\gamma_{SL},\gamma_{SG}$} & \parbox[t]{0.21\textwidth}{\raggedright support reaction + reversible configurational migration} & \parbox[t]{0.23\textwidth}{\raggedright Young law against the held plane}\\
\bottomrule
\end{tabular}
\end{table}

\subsection{Geometric reference and observational resolution of contact angles}

A liquid angle measured against the undeformed support plane need not equal the angle measured against a ridge flank. Local ridge measurements probe solid surface stresses \cite{StyleEtAl2013}. Lubbers et al. identified distinct microscopic and macroscopic crossovers involving the molecular interaction length, drop radius and elastocapillary length \cite{LubbersEtAl2014}; force transmission and regularisation further affect the geometry near the contact line \cite{MarchandSoft2012,DervauxLimat2015}. These distinctions motivate the ordered spatial limits used here, without selecting a universal molecular cutoff length.

Prestrain can make an initially isotropic surface mechanically anisotropic \cite{HeydenEtAl2021}, and asymmetric wet/dry Shuttleworth responses can alter horizontal deformation and contact-line behaviour \cite{HenkelEtAl2022}. Surface potentials, held stretch and one-sided tip states must therefore be specified separately. The constitutive no-relaxation passage in \cref{sec:viscoelastic-solid} also has a different purpose from the dissipative moving-ridge models listed above: it identifies an elastic equilibrium model rather than a dynamic contact-angle law.

Apparent angle changes with contact-line radius can arise from state-dependent interfacial potentials as well as a retained line energy. A radius dependence alone therefore does not identify a constant material line tension \cite{WardWu2008,LiuWangZhang2013}.

\section{Conclusions}
\label{sec:conclusion}

The passage from a three-fluid junction to a rigid wetting support changes both the constitutive state and the admissible motions. Permanent affine memory supplies an elastic bulk response. Independently specified material-surface energies distinguish wet--dry conversion from surface deformation. On the compliant solid, spatial displacement and migration give separate equilibrium conditions, with compatible one-sided variations at a ridge.

Rigidity constrains solid material displacement while reversible wet--dry migration can remain possible. Restricting the boundary variation to that surviving motion yields Young's law. The capillary load remains in the mechanical balance through the solid surface stresses and the support reaction, including its normal component and any required tangential component. Contact-line migration generates no reaction work when the solid material remains fixed.

The finite-strain almost-minimiser result and the separate quadratic stationary calculation substantiate this mechanism within explicit classes. They also identify the missing requirements in a general passage: convergence of the relevant minimality property or migration variation, and separate control of the retained mechanics. These distinctions provide a clear basis for using Young's law as the outer equilibrium condition while retaining the forces carried by the support.

\section*{Data accessibility}
No experimental or external research data were used. The equations and proofs are contained in the article and its appendices. Python code and a machine-readable table reproducing Table~\ref{tab:quad-numerics} are supplied as electronic supplementary material.

\section*{Competing interests}
The author declares no competing interests.

\section*{Funding}
The author declares that no funds, grants, or other support were received during the preparation of this manuscript.

\section*{Use of artificial intelligence}
OpenAI ChatGPT assisted with restructuring and wording, notation and mathematical consistency checks, and preparing the numerical code for the quadratic-continuum benchmark. The author directed the revisions and retains responsibility for the scientific content and its verification.

\clearpage
\appendix
\numberwithin{equation}{section}
\section{Constitutive estimate along a prescribed motion}
\label{sec:transport-estimate}

\subsection{Finite-time transport estimate}

Prescribe a common smooth incompressible motion on $0\le t\le T$, where $t$ is time and $T$ is a fixed observation duration independent of the relaxation time $\lambda$. Take $\vct F(0)=\vct I$ and initial conformation $\vct A_\lambda(0)=\vct A_0$, with $\vct A_0$ symmetric positive definite. Define the pulled-back conformation $\vct C_\lambda$ and inverse-metric tensor $\vct H$ along each trajectory by
\begin{equation}
\vct C_\lambda=\vct F^{-1}\vct A_\lambda\vct F^{-\mathsf T},
\qquad \vct H(t)=\vct F^{-1}(t)\vct F^{-\mathsf T}(t).
\end{equation}
Here $\vct F^{-\mathsf T}$ denotes the transpose of $\vct F^{-1}$. Equation~\eqref{eq:oldroyd-A} becomes a linear matrix equation, whose solution is
\begin{equation}
\dot{\vct C}_\lambda=-\lambda^{-1}(\vct C_\lambda-\vct H),
\qquad
\vct C_\lambda(t)=e^{-t/\lambda}\vct A_0
+\frac1\lambda\int_0^t e^{-(t-s)/\lambda}\vct H(s)\dd s.
\label{eq:exact-pullback}
\end{equation}
The dot denotes differentiation along the trajectory, and $s$ is the integration time. Both terms preserve positive definiteness. The no-relaxation conformation for the same motion is $\vct A_\infty=\vct F\vct A_0\vct F^{\mathsf T}$. Let $\|\cdot\|$ denote the operator norm, and let $M_F$ and $M_-$ be uniform bounds on the norms of $\vct F$ and $\vct F^{-1}$, respectively, on $[0,T]$. Subtracting $\vct A_0$ inside the integral gives
\begin{equation}
\begin{aligned}
\sup_{0\le t\le T}\|\vct A_\lambda-\vct A_\infty\|
&\le M_F^2\bigl(1-e^{-T/\lambda}\bigr)
\bigl(M_-^2+\|\vct A_0\|\bigr)\\
&\le \frac{T}{\lambda}M_F^2\bigl(M_-^2+\|\vct A_0\|\bigr).
\end{aligned}
\label{eq:finite-time-bound}
\end{equation}
This quantifies the affine constitutive correspondence \cite{SnoeijerEtAl2020Elasticity}. At fixed $G_S$, multiplication by $G_S$ gives the corresponding polymer-stress estimate. Uniform positive lower eigenvalue bounds imply convergence of the conformation energy in \eqref{eq:oldroyd-energy}.

In a coupled problem, the motion itself depends on $\lambda$. The same calculation compares $\vct A_\lambda$ with the elastic tensor transported by that motion; identifying the limiting motion requires additional uniform bounds, convergence and contact-line compatibility. The estimate is not uniform as $T\to\infty$.

For a general initial reference placement, the no-relaxation solution is
\begin{equation}
\vct A=\vct F\vct D_0\vct F^{\mathsf T},\qquad
\vct D_0(\vct X)=\vct F(0)^{-1}\vct A(0)\vct F(0)^{-\mathsf T},
\label{eq:A-D0}
\end{equation}
where $\vct X$ is a reference material point and $\vct D_0$ is time independent. Positivity of this memory tensor does not establish a globally compatible stress-free Euclidean reference. The relaxed homogeneous neo-Hookean identification uses the matched case.

\subsection{Stored energy and dissipation}

For the conformation energy density $W_p$ in \eqref{eq:oldroyd-energy}, contraction of \eqref{eq:oldroyd-A} with $(G_S/2)(\vct I-\vct A^{-1})$ gives
\begin{equation}
\frac{DW_p}{Dt}=G_S(\vct A-\vct I):\vct D_3-\mathcal D_{\rm rel},
\label{eq:oldroyd-energy-identity}
\end{equation}
where the relaxation dissipation density is
\begin{equation}
\mathcal D_{\rm rel}
=\frac{G_S}{2\lambda}\tr(\vct A+\vct A^{-1}-2\vct I)\ge0.
\label{eq:oldroyd-relax-diss}
\end{equation}
The inequality follows from $a+a^{-1}-2\ge0$ for each positive eigenvalue $a$ of $\vct A$. The solvent contributes $2\eta_s\vct D_3:\vct D_3\ge0$. On a fixed interval with bounded $\vct A$ and $\vct A^{-1}$, the relaxation dissipation tends to zero as $\lambda^{-1}$; for matched incompressible memory, $\det\vct B=1$ and $W_p$ reduces to $W_{NH}$. These statements concern the constitutive equations and do not assert long-time convergence of the coupled dynamics.

\section{Supporting finite-strain passage for global almost minimisers}
\label{sec:nonlinear-continuum}

This appendix proves Young recovery and uniform convergence of the energy minimised over solid placements. The model has a clamped solid with unchanged transverse coordinate (plane strain), the surface energies specified below, and symmetric circular liquid interfaces. The argument uses global almost minimisers, whose energies approach the global infimum within a vanishing tolerance. It requires neither exact minimisers nor uniform ridge-tip strain bounds.

\subsection{An exact bulk--surface energy and exact liquid geometry}

Let $(X,Z)$ be the horizontal and vertical reference coordinates in $\Omega=(-L,L)\times(-H,0)$, a body of fixed half-width $L>0$ and depth $H>0$. All energies are per unit unchanged transverse length. An admissible placement $\vct\chi:\overline\Omega\to\R^2$ is an orientation-preserving homeomorphism onto its image, belongs to $H^1(\Omega;\R^2)$,\footnote{A homeomorphism is a continuous bijection with continuous inverse. The space $H^s$ is the $L^2$-based Sobolev space of order $s$; in particular, $H^1$ requires square-integrable values and first weak derivatives. The boundary trace, meaning the boundary restriction in the Sobolev sense, lies in $H^{1/2}$. On an interval, $H^1_0$ imposes zero endpoint values. The $L^2$ norm is the square root of the integral of squared magnitude, and the $L^\infty$ norm is the essential supremum.} and equals the identity on the sides and bottom. Its deformation gradient $\vct F=\nabla_{(X,Z)}\vct\chi$ satisfies $\det\vct F=1$ almost everywhere. No uniform pointwise gradient bound is imposed. Write its top trace in terms of the horizontal displacement $u$ and vertical displacement $h$ as
\begin{equation}
\vct r(X)=(X+u(X),h(X)),\qquad
\lambda_s=|\vct r'|,\qquad u,h\in H^1_0(-L,L).
\label{eq:nl-top-trace}
\end{equation}
Here $\lambda_s$ is the surface stretch, a prime differentiates with respect to the displayed scalar argument (thus $\vct r'=\dd\vct r/\dd X$), and $|\cdot|$ is the Euclidean norm for vectors or the square root of the sum of squared entries for matrices (the Frobenius norm). We impose reflection symmetry ($u$ odd, $h$ even) and $1+u'>0$ almost everywhere, so that the top is a graph. The wet reference interval has half-width $a\in I=[a_-,a_+]\subset(0,L)$, with fixed endpoints $a_-<a_+$. Its current half-width is $b=a+u(a)>0$, and its endpoint height is $h_e=h(a)$. The two spatial endpoints are $(\pm b,h_e)$.

The liquid has fixed cross-sectional area $V>0$, constant tension $\gamma:=\gamma_{LG}>0$, no gravity, and a simple circular liquid--gas arc meeting the solid only at these endpoints. Admissibility includes absence of solid/liquid self-intersection. The area $S$ of the circular segment above the endpoint chord is \emph{not} generally $V$:
\begin{equation}
 S=V+\int_{-a}^{a}h(X)(1+u'(X))\dd X-2b h_e>0.
\label{eq:nl-exact-volume}
\end{equation}
The factor $1+u'$ converts reference to current horizontal length. Let $\Theta$ be the spatial liquid-interface angle relative to the fixed horizontal, measured through the circular segment. Define the dimensionless circular-segment function $Q$ and liquid--gas capillary energy $\mathcal C$ by
\begin{equation}
 Q(\Theta)=\frac{\Theta-\sin\Theta\cos\Theta}{\sin^2\Theta},\qquad
 S=b^2Q(\Theta),\qquad
 \mathcal C(b,S)=2\gamma b\frac{\Theta}{\sin\Theta},\quad 0<\Theta<\pi.
\label{eq:nl-cap-energy}
\end{equation}
The angle $\Theta$ is distinct from the angle relative to a one-sided ridge tangent. Since $Q'>0$, $\Theta(b,S)$ is uniquely defined and smooth for $b,S>0$. Direct differentiation gives
\begin{equation}
 Q'(\Theta)=\frac{2(\sin\Theta-\Theta\cos\Theta)}{\sin^3\Theta},\qquad
 \mathcal C_b\big|_S=2\gamma\cos\Theta,\qquad
 \mathcal C_S\big|_b=\frac{\gamma\sin\Theta}{b}=:p.
\label{eq:nl-cap-derivatives}
\end{equation}
Here the subscripts denote partial derivatives at the indicated fixed variable, and $p$ is the capillary pressure relative to the gas. All subsequent claims concern this circular-interface class.

For the wet ($i=SL$) and dry ($i=SG$) solid surfaces, choose the energy per unit reference area $\psi_i$ from the family in \cref{sec:benchmark}:
\begin{equation}
 \psi_i(\lambda_s)=\gamma_{0,i}+\Upsilon_{0,i}(\lambda_s-1)
             +\frac{K_i}{2}(\lambda_s-1)^2,
 \qquad i\in\{SL,SG\},\quad \Upsilon_{0,i}\ge K_i>0.
\label{eq:nl-common-law}
\end{equation}
The constants $\gamma_{0,i}=\psi_i(1)$ are the undeformed equilibrium surface energies, $\Upsilon_{0,i}$ the reference surface stresses, and $K_i$ the surface stretch moduli. Taking $\gamma_{0,i}>\Upsilon_{0,i}-K_i/2$ also makes $\psi_i>0$ for $\lambda_s>0$. The positive stretch moduli exclude the liquid-like law $\psi_i(\lambda_s)=\gamma_{0,i}\lambda_s$; that case requires a different estimate. The surface stresses $\Upsilon_i=\psi_i'$ are tensile, while the material conjugates $H_i=\psi_i-\lambda_s\psi_i'$ measure the energy variation under material endpoint migration:
\begin{equation}
 \Upsilon_i=\Upsilon_{0,i}+K_i(\lambda_s-1),\qquad
 H_i=\gamma_{0,i}-\Upsilon_{0,i}+\frac{K_i}{2}(1-\lambda_s^2).
\label{eq:nl-common-conjugates}
\end{equation}
These laws do not identify stress and surface energy: the current-area energy is $\gamma_i(\lambda_s)=\psi_i(\lambda_s)/\lambda_s$. Write $\Delta\gamma=\gamma_{0,SG}-\gamma_{0,SL}$ and $\Delta\Upsilon_0=\Upsilon_{0,SG}-\Upsilon_{0,SL}$. In the earlier notation \eqref{eq:benchmark-law}, the correspondence is $k_i=K_i$ and $b_i=\Upsilon_{0,i}-\gamma_{0,i}$. Only the reference state $\lambda_s=1$ has the surface-energy difference $\Delta\gamma$; it must not be used indiscriminately at a strained tip.

Let $G=G_S>0$ denote the solid shear modulus. After subtracting the constant dry baseline $2L\gamma_{0,SG}$, the total energy $\mathcal E_G$ decomposes into bulk energy $B_G$, a nonnegative surface remainder $\mathscr P_a$, and capillary and endpoint terms:
\begin{align}
 \mathcal E_G(\vct\chi,a)
 &=B_G(\vct\chi)+\mathscr P_a(\vct r)+\mathcal C(b,S)
            -2a\Delta\gamma-2\Delta\Upsilon_0u(a),
 \label{eq:nl-energy}\\
 B_G&=\frac G2\int_\Omega(|\vct F|^2-2)\dd X\dd Z,\nonumber\\
 \mathscr P_a&=\sum_{i\in\{SL,SG\}}\int_{I_i(a)}
 \left[\Upsilon_{0,i}(\lambda_s-1-u')+\frac{K_i}{2}(\lambda_s-1)^2\right]\dd X,
 \label{eq:nl-positive-energy}
\end{align}
where $I_{SL}=(-a,a)$ and $I_{SG}=(-L,-a)\cup(a,L)$. The final term follows from
$\sum_i\int_{I_i}\Upsilon_{0,i}u'\dd X=-2\Delta\Upsilon_0u(a)$ by symmetry and the clamps. It belongs to the surface functional and will be controlled by the trace estimate below. The bulk term is the finite-strain incompressible neo-Hookean energy. The clamps do no work; external dead loads, line energy, and independent pinning barriers are absent. The identity placement with the corresponding flat circular cap is assumed admissible for every $a\in I$.

\subsection{Coercivity and the critical trace estimate}

The boundary space $H^{1/2}$ is critical here because its norm alone does not control point values. The surface energy supplies an additional $H^1$ estimate. Let $x(X)=X+u(X)$, so that $x'=1+u'$ is the horizontal component of the surface tangent, and put $c_*:=\min_i\{K_i,\Upsilon_{0,i}\}>0$. With $\vct e_x=(1,0)$ the horizontal unit vector, pointwise $\lambda_s\ge x'$ and
$|\vct r'-\vct e_x|^2=(\lambda_s-1)^2+2(\lambda_s-x')$. Hence, with no small-slope assumption,
\begin{equation}
 \mathscr P_a\ge\frac{c_*}{2}\int_{-L}^{L}|\vct r'-\vct e_x|^2\dd X
 =\frac{c_*}{2}\int_{-L}^{L}(u'^2+h'^2)\dd X\ge0,
 \qquad a\in I.
\label{eq:nl-surface-coercivity}
\end{equation}
This bound is uniform in the moving partition and is the required surface coercivity, namely quadratic control of the surface displacement gradient. If $K_{SL}=K_{SG}=:K_s$ and $\Upsilon_{0,SL}=\Upsilon_{0,SG}=: \Upsilon_0$, then $\mathscr P_a$ is independent of $a$ and can be written as
$\mathscr P=\frac{K_s}{2}\int_{-L}^L|\vct r'-\vct e_x|^2\dd X+(\Upsilon_0-K_s)\int_{-L}^L(\lambda_s-1)\dd X$, with the latter integral nonnegative by the endpoint clamps.
Write $SO(2)$ for the two-dimensional rotation matrices and $\operatorname{dist}$ for distance in the Frobenius norm. For $\det\vct F=1$, $\operatorname{dist}^2(\vct F,SO(2))\le|\vct F|^2-2$. The geometric-rigidity theorem \cite{FJM2002} controls the distance of a placement from a single rigid motion by this integrated gradient error. The fixed-boundary Poincar\'e inequality controls displacement by its gradient, giving
\begin{equation}
 \|\vct\chi-\operatorname{id}\|_{H^1(\Omega)}
 \le C\sqrt{B_G/G}.
\label{eq:nl-rigidity}
\end{equation}
Here $\operatorname{id}$ is the identity placement. The symbols $C$ and its indexed variants denote finite positive constants independent of $G$; their values may change between estimates. Rigidity first supplies a single rotation and translation. Their discrepancy from the identity is bounded by the same error because the boundary values are fixed on a nontrivial segment. The constant depends on the fixed reference geometry. This step would not hold unchanged for an unclamped body or a simultaneously degenerating domain.

Assume $|\Delta\gamma|<\gamma$. Define the flat Young angle $\theta_Y$ and contact half-width $a_Y$, the dimensionless stiffness $\beta$, a vanishing rate parameter $\eta_\beta$, and the dimensionless bulk and surface energy $e$ by
\begin{equation}
 \theta_Y=\arccos(\Delta\gamma/\gamma),\quad
 a_Y=\sqrt{V/Q(\theta_Y)},\quad
 \beta=\frac{G a_Y}{\gamma},\quad
 \eta_\beta=\frac{1+\log\beta}{\beta},\quad
 e=\frac{B_G+\mathscr P_a}{\gamma a_Y}.
\label{eq:nl-scales}
\end{equation}
Choose $I$ with $a_Y$ in its interior. All constants below may depend on $L/a_Y$, $H/a_Y$, $a_-/a_Y$, $a_+/a_Y$, $c_*/\gamma$, $\Upsilon_{0,i}/\gamma$, $K_i/\gamma$, and $\theta_Y$, but not on $\beta\ge2$.

\begin{lemma}[Uniform displacement from two trace norms]
\label{lem:nl-trace}
For every admissible placement,
\begin{equation}
 \frac{\|(u,h)\|_{L^\infty(-L,L)}}{a_Y}
 \le C\sqrt{e\eta_\beta}.
\label{eq:nl-trace-bound}
\end{equation}
\end{lemma}
\begin{proof}
Rescale both reference coordinates by $a_Y$ and write $w=(u,h)/a_Y$ on the resulting interval. The bulk trace theorem, which controls the boundary $H^{1/2}$ norm by the interior $H^1$ norm, and \eqref{eq:nl-rigidity} bound $\|w\|_{H^{1/2}}$ by $C\sqrt{e/\beta}$. Surface coercivity and the zero endpoint values bound $\|w\|_{H^1}$ by $C\sqrt e$. Extend $w$ to a larger periodic interval by reflection and multiplication by a smooth cutoff; this fixed extension is bounded in both spaces. Splitting its Fourier series at the integer mode $N\ge2$ and applying Cauchy--Schwarz gives
\begin{equation}
 \|w\|_\infty\le C\left[
 \sqrt{\log(2+N)}\,\|w\|_{H^{1/2}}
       +N^{-1/2}\|w\|_{H^1}\right].
\label{eq:nl-fourier-split}
\end{equation}
The sums over integer Fourier modes $n$ satisfy $\sum_{|n|\le N}(1+|n|)^{-1}\le C\log(2+N)$ and $\sum_{|n|>N}n^{-2}\le C/N$. Choosing $N=\lceil\beta\rceil$ proves the estimate.
\end{proof}

\subsection{Young recovery without a migration-derivative hypothesis}

For the flat geometry let $\theta_0(a)=\Theta(a,V)$ be the cap angle and $\mathcal E_W$ the total energy after subtraction of the dry baseline:
\begin{equation}
 \mathcal E_W(a)=\mathcal C(a,V)-2a\Delta\gamma,\qquad
 \mathcal E_W'=2(\gamma\cos\theta_0-\Delta\gamma),\qquad
 \mathcal E_W''=-2\gamma\sin\theta_0\,\theta_0'>0,
\label{eq:nl-flat-convexity}
\end{equation}
where
\begin{equation}
 \theta_0'(a)=-\frac{(\theta_0-\sin\theta_0\cos\theta_0)\sin\theta_0}
 {a(\sin\theta_0-\theta_0\cos\theta_0)}<0.
\label{eq:nl-flat-angle-derivative}
\end{equation}
Thus $a_Y$ is the unique flat minimiser, and $\mathcal E_W$ is uniformly strongly convex on $I$: its second derivative has a positive lower bound there.

\begin{theorem}[Finite-strain Young limit in the stated cap class]
\label{thm:nl-young}
Under the preceding hypotheses, let $(\vct\chi_G,a_G)$ be admissible global almost minimisers with dimensionless energy tolerance $\delta_\beta$:
\begin{equation}
 \mathcal E_G(\vct\chi_G,a_G)
 \le\inf_{\vct\chi,\,a\in I}\mathcal E_G(\vct\chi,a)
            +\gamma a_Y\delta_\beta,\qquad
 \delta_\beta\ge0,\quad\delta_\beta\longrightarrow0.
\label{eq:nl-almost-minimizers}
\end{equation}
Use the subscript $G$ for the top trace $\vct r_G$, displacement fields $u_G,h_G$, current half-width $b_G$, circular-segment area $S_G$, and dimensionless energy $e_G$ associated with this pair. Let $\Theta_G=\Theta(b_G,S_G)$ be its spatial cap angle, with $S_G$ given by \eqref{eq:nl-exact-volume}. For all sufficiently large $\beta$,
\begin{align}
 e_G+\frac{\mathcal E_W(a_G)-\mathcal E_W(a_Y)}{\gamma a_Y}
       &\le C\eta_\beta+2\delta_\beta,
\label{eq:nl-energy-rate}\\
 \frac{\|(u_G,h_G)\|_\infty}{a_Y}
       &\le C\sqrt{\eta_\beta(\eta_\beta+\delta_\beta)},
\label{eq:nl-displacement-rate}\\
 \frac{|a_G-a_Y|}{a_Y}+|\Theta_G-\theta_Y|
       &\le C\sqrt{\eta_\beta+\delta_\beta}.
\label{eq:nl-angle-rate}
\end{align}
In particular, the spatial angle satisfies Young's equation in the stiff limit. The theorem applies to exact global minimisers whenever they exist, and to arbitrary admissible almost-minimizing sequences with the displayed vanishing error.
\end{theorem}
\begin{proof}
Comparison with the identity placement at $a_Y$ gives
$\mathcal E_G(\vct\chi_G,a_G)\le\mathcal E_W(a_Y)+\gamma a_Y\delta_\beta$.
Because $\mathcal C>0$ and $a_G\in I$, the comparison and Lemma~\ref{lem:nl-trace} first give
$e_G\le C_0+C_1\sqrt{e_G\eta_\beta}+\delta_\beta$; the square-root term bounds $2\Delta\Upsilon_0u_G(a_G)/(\gamma a_Y)$. The product estimate $C_1\sqrt{e_G\eta_\beta}\le e_G/2+C_1^2\eta_\beta/2$ gives a uniform bound on $e_G$. The trace lemma then makes the top displacement uniformly small. Surface coercivity also bounds the current surface length $\int_{-L}^L\lambda_s\dd X$. Hence, by \eqref{eq:nl-exact-volume},
\begin{equation}
 \frac{|b_G-a_G|}{a_Y}+\frac{|S_G-V|}{a_Y^2}
 \le C\frac{\|(u_G,h_G)\|_\infty}{a_Y}.
\label{eq:nl-geometry-control}
\end{equation}
For example, the integral containing $h(1+u')$ is bounded by $\|h\|_\infty\int\lambda_s\dd X$; no pointwise bound on $u'$ is needed. For large $\beta$, $(b_G,S_G)$ therefore lies in a fixed compact subset of $(0,\infty)^2$ near $I\times\{V\}$. Smoothness of $\mathcal C$ and \eqref{eq:nl-trace-bound} imply
\begin{equation}
 \frac{|\mathcal C(b_G,S_G)-\mathcal C(a_G,V)|}{\gamma a_Y}
 \le C\sqrt{e_G\eta_\beta}.
\label{eq:nl-cap-comparison}
\end{equation}
Define the normalised excess flat energy $d_G=[\mathcal E_W(a_G)-\mathcal E_W(a_Y)]/(\gamma a_Y)\ge0$. The comparison inequality and the trace bound on the reference-stress point term give
$e_G+d_G\le C\sqrt{e_G\eta_\beta}+\delta_\beta$.
The same product estimate yields $e_G/2+d_G\le C\eta_\beta+\delta_\beta$, and hence \eqref{eq:nl-energy-rate} after changing the constant. The trace lemma gives \eqref{eq:nl-displacement-rate}. Strong convexity in \eqref{eq:nl-flat-convexity} bounds $|a_G-a_Y|/a_Y$ by $C\sqrt{d_G}$; smoothness of $\Theta(b,S)$ and \eqref{eq:nl-geometry-control} give the remaining term in \eqref{eq:nl-angle-rate}.
\end{proof}

The infimum in \eqref{eq:nl-almost-minimizers} is finite: \eqref{eq:nl-energy} and the trace lemma give a lower bound $\gamma a_Y(e-C\sqrt{e\eta_\beta}-C_0)$, and flat competitors give a finite upper bound. Positive-tolerance almost minimisers therefore exist by the definition of infimum. Taking, for example, $\delta_\beta=\beta^{-2}$ gives the displayed rates even if no exact minimiser exists.

\begin{corollary}[Uniform convergence of the reduced energy]
\label{cor:nl-reduced-energy}
Let $m_G(a)=\inf_{\vct\chi}\mathcal E_G(\vct\chi,a)$, where the infimum is over admissible placements at the fixed partition $a\in I$. For all sufficiently large $\beta$,
\begin{equation}
 \mathcal E_W(a)-C\gamma a_Y\eta_\beta
 \le m_G(a)\le\mathcal E_W(a)
 \qquad\text{uniformly for }a\in I.
\label{eq:nl-uniform-reduced-energy}
\end{equation}
\end{corollary}
\begin{proof}
The identity placement gives the upper bound. For the lower bound it suffices to consider placements with $\mathcal E_G(\vct\chi,a)\le\mathcal E_W(a)$. Since $\mathcal E_W$ is bounded on $I$ and $\mathcal C>0$, \eqref{eq:nl-energy} and Lemma~\ref{lem:nl-trace} give $e\le C_0+C_1\sqrt{e\eta_\beta}$, uniformly in $a$. This inequality bounds $e$ uniformly. The exact-area and cap comparisons \eqref{eq:nl-geometry-control}--\eqref{eq:nl-cap-comparison} consequently apply to every such placement with constants independent of $a$. Bounding the reference-stress term in the same way gives
\[
 \frac{\mathcal E_G(\vct\chi,a)-\mathcal E_W(a)}{\gamma a_Y}
 \ge e-C\sqrt{e\eta_\beta}\ge-C'\eta_\beta.
\]
All other placements already satisfy the lower bound. Taking the infimum proves the claim without requiring an attaining placement.
\end{proof}

\paragraph{Extension to nonquadratic surface laws.}
The proof uses only $\gamma_{0,i}=\psi_i(1)$, $\Upsilon_{0,i}=\psi_i'(1)>0$, and a quadratic lower bound on the energy above its tangent at unit stretch. Specifically, for every stretch $\lambda>0$ suppose that a fixed modulus $k_*>0$ satisfies
\begin{equation}
 \psi_i(\lambda)-\gamma_{0,i}-\Upsilon_{0,i}(\lambda-1)
 \ge\frac{k_*}{2}(\lambda-1)^2\qquad(\lambda>0),\quad k_*>0.
\label{eq:nl-general-surface-coercivity}
\end{equation}
Replace $K_i(\lambda-1)^2/2$ in $\mathscr P_a$ by the left-hand remainder in \eqref{eq:nl-general-surface-coercivity}. Then \eqref{eq:nl-energy} stays exact and \eqref{eq:nl-surface-coercivity} holds with $c_*=\min\{k_*,\Upsilon_{0,SL},\Upsilon_{0,SG}\}$. Every subsequent estimate is unchanged for finite-energy admissible placements. Both conclusions therefore extend to this nonquadratic surface class. Local convexity at unit stretch is insufficient when tip strains remain uncontrolled.

\subsection{Scope of the finite-strain proof}

The theorem and corollary establish the minimiser route without \eqref{eq:first-variation-convergence}. They also give $\|\vct\chi_G-\operatorname{id}\|_{H^1(\Omega)}\to0$ and $\|\vct r_G'-\vct e_x\|_{L^2(-L,L)}\to0$ along the almost minimisers. These estimates allow strain concentration in a shrinking interval and provide no uniform pointwise bound on the actual surface stresses. If separate regularity and zero bulk material force justify the local corner condition, the quantities in \eqref{eq:nl-common-conjugates} must satisfy $H_{SL}=H_{SG}$. With $\lambda_{SL}$ and $\lambda_{SG}$ the respective one-sided stretches at the ridge tip, the common-response case gives $K_s(\lambda_{SG}^2-\lambda_{SL}^2)/2=\Delta\gamma$, rather than equal tip stretches.

Arbitrary nonlinear stationary branches, their migration derivatives, and their reaction-stress limits remain outside this result. Prestretched or variable domains, noncoercive surface energies, line energies, and other loading ensembles require a new comparison analysis. Appendix~\ref{sec:quadratic-continuum} computes stationary derivatives and reaction measures for a separately defined quadratic half-space model.

\section{Supporting stationary and reaction calculation in a half-space}
\label{sec:quadratic-continuum}

This calculation resolves the stationary migration derivative and the reaction in an explicitly solvable model. Let $x$ be the horizontal coordinate and $z<0$ the depth coordinate in an incompressible \emph{linear-elastic} half-space of shear modulus $G=G_S$. All fields are independent of the transverse coordinate, with no displacement in that direction; energies are per unit transverse length. The displacement is $\vct v=(v_x,v_z)$, with tangential and normal surface traces $u(x)=v_x(x,0)$ and $h(x)=v_z(x,0)$. We use the quadratic expansion of \eqref{eq:nl-common-law}, with common wet/dry surface stretch moduli $K_{SL}=K_{SG}=K_s>0$ and reference surface stresses $\Upsilon_{0,SL}=\Upsilon_{0,SG}=\Upsilon_0>0$, and retain the capillary energy to \emph{first order} in displacement about a flat circular cap. The reference surface energies may still differ. The resulting quadratic functional is solved exactly; its half-space kernels and surface-elastic regularisation are established ingredients \cite{LubbersEtAl2014,StyleXu2018}. It is a separate model, not a uniform linearisation of the finite-strain ridge in Appendix~\ref{sec:nonlinear-continuum}.

\subsection{Complete capillary loading and energy elimination}

Retain the reference wet half-width $a$, fixed liquid cross-sectional area $V$, and liquid--gas tension $\gamma=\gamma_{LG}$ from Appendix~\ref{sec:nonlinear-continuum}. The flat-cap angle is $\theta_0(a)=\Theta(a,V)$, where $\Theta(b,S)$ is the circular-segment angle for half-width $b$ and segment area $S$ in \eqref{eq:nl-cap-energy}. Define the normal line-load amplitude $q$ at each contact line, the tangential amplitude $f$ at the right line, and the liquid pressure $p$ by
\begin{equation}
 q=\gamma\sin\theta_0,\qquad f=-\gamma\cos\theta_0,\qquad p=q/a.
\label{eq:quad-load-amplitudes}
\end{equation}
The first variation $\delta\mathcal C$ of the liquid--gas capillary energy, using the exact geometry \eqref{eq:nl-exact-volume} at $u=h=0$, is
\begin{equation}
 \delta\mathcal C=\gamma\cos\theta_0\,[u(a)-u(-a)]
 -q[h(a)+h(-a)]+p\int_{-a}^{a}h(x)\dd x.
\label{eq:quad-capillary-expansion}
\end{equation}
Thus the tangential and normal capillary load distributions $t_x,t_z$ conjugate to $u,h$ are
\begin{equation}
 t_x=f[\delta(x-a)-\delta(x+a)],\qquad
 t_z=q[\delta(x-a)+\delta(x+a)]-p\,\mathbf1_{(-a,a)}.
\label{eq:quad-loads}
\end{equation}
Here $\delta(x-a)$ is the unit Dirac mass at $x=a$, and $\mathbf1_{(-a,a)}$ equals one on the wet interval and zero elsewhere. Both loads have zero resultant. Omitting the liquid pressure would produce an energy divergence from arbitrarily long wavelengths. The common linear surface-prestress term integrates to zero for decaying $u$, leaving $K_su'^2/2+\Upsilon_0h'^2/2$, where primes on surface displacements denote differentiation with respect to $x$.

For any surface field $v$, use the Fourier transform $\widehat v(k)=\int_\R v(x)e^{-ikx}\dd x$, where $k$ is the horizontal wave number and $i=\sqrt{-1}$. Eliminating the decaying incompressible bulk field gives the tangential and normal surface stiffness multipliers $D_x,D_z$:
\begin{equation}
 D_x(k)=2G|k|+K_s k^2,\qquad D_z(k)=2G|k|+\Upsilon_0 k^2.
\label{eq:quad-stiffness}
\end{equation}
The incompressible Young modulus is $3G$. The factor $2G|k|$ follows from the bulk solution in \cref{app:spectral}, consistently with Refs.~\cite{LubbersEtAl2014,StyleXu2018}. With $\mathcal E_W(a)$ the flat-wall energy of Appendix~\ref{sec:nonlinear-continuum}, define the quadratic energy $\mathcal E_G^{\rm q}$ by
\begin{equation}
 \mathcal E_G^{\rm q}(a,u,h)=\mathcal E_W(a)
 +\frac12\int_\R\bigl(D_x|\widehat u|^2+D_z|\widehat h|^2\bigr)\frac{\dd k}{2\pi}
 -\langle t_x,u\rangle-\langle t_z,h\rangle.
\label{eq:quad-energy}
\end{equation}
The brackets $\langle t,v\rangle$ pair a load distribution with a displacement. The final two terms are the first-order capillary energy in \eqref{eq:quad-capillary-expansion}, so no additional capillary work is added. We minimize over surface traces for which the quadratic integral is finite, identifying displacements that differ by constants and selecting the decaying representative. The zero-resultant loads are bounded linear functionals on this space; $K_s,\Upsilon_0>0$ regularise the short wavelengths. In particular, $K_s>0$ is needed for a nonzero tangential point load to have finite energy. A liquid-like law with $K_s=0$ is therefore outside this model.

Completion of the square yields $\widehat u=\widehat t_x/D_x$, $\widehat h=\widehat t_z/D_z$, and the minimised energy $\mathcal E_{{\rm red},G}^{\rm q}$. Its correction to $\mathcal E_W$ is denoted by $r_G$:
\begin{equation}
 \mathcal E_{{\rm red},G}^{\rm q}(a)=\mathcal E_W(a)+r_G(a),\qquad
 r_G=-\frac{q^2J_n(\epsilon_\Upsilon)+f^2J_t(\epsilon_K)}{\pi G},\qquad
 \epsilon_s=\frac{s}{2Ga}.
\label{eq:quad-elimination}
\end{equation}
Here $s$ is either $K_s$ or $\Upsilon_0$, and $\epsilon_s$ is the corresponding surface-elastic length $s/(2G)$ divided by the reference half-width $a$; $\epsilon_K$ and $\epsilon_\Upsilon$ denote these two values. The tangential and normal dimensionless kernels $J_t,J_n$ use the integration variable $t=ka$ for $k>0$:
\begin{equation}
 J_t(\epsilon)=\int_0^\infty\frac{\sin^2t}{t(1+\epsilon t)}\dd t,\qquad
 J_n(\epsilon)=\int_0^\infty\frac{(\cos t-\sin t/t)^2}{t(1+\epsilon t)}\dd t.
\label{eq:quad-kernels}
\end{equation}
Although the stored elastic energy is positive, $r_G$ is negative: it includes the displacement-dependent capillary energy, and the minimised sum equals minus one-half of the load--displacement pairing. Replacing $r_G$ by the stored energy would reverse the sign of the migration correction.

\subsection{Derivative control and the common-stress cancellation}

For the kernel index $j\in\{t,n\}$, put $D_j(\epsilon)=-\epsilon J_j'(\epsilon)$, where the prime denotes differentiation with respect to $\epsilon$. These $D_t,D_n$ are kernel derivatives, distinct from the stiffnesses $D_x,D_z$. The identities proved in \cref{app:spectral} give, as $\epsilon\downarrow0$,
\begin{align}
 J_t(\epsilon)&=\tfrac12[\log(2/\epsilon)+\gamma_E]+O(\epsilon^2),\nonumber\\
 J_n(\epsilon)&=\tfrac12[\log(2/\epsilon)+\gamma_E]-\tfrac12
                         +O(\epsilon^2|\log\epsilon|),\qquad D_j\longrightarrow\tfrac12,
\label{eq:quad-kernel-asymptotics}
\end{align}
where $\gamma_E$ is Euler's constant. Here $O(g)$ denotes a term whose ratio to $g$ stays bounded in the indicated limit, and $o(g)$ denotes a term whose ratio tends to zero; these bounds are uniform on the stated parameter intervals. The $D_j$ and $\epsilon D_j'$ are uniformly bounded. Differentiating \eqref{eq:quad-elimination} \emph{after} scaling the integrals to $t=ka$ gives the exact formula below, where primes on $r_G,q,f$ denote differentiation with respect to $a$:
\begin{equation}
 r_G'=-\frac1{\pi G}\left[
 2qq'(J_n-J_t)+\frac{q^2D_n+f^2D_t}{a}\right].
\label{eq:quad-derivative}
\end{equation}
Here each kernel has its corresponding $\epsilon_s$ argument. The simplification uses $q^2+f^2=\gamma^2$, hence $qq'+ff'=0$. In particular,
\begin{equation}
 r_G'=-\frac1{\pi G}\left[
 qq'\left(\log\frac{K_s}{\Upsilon_0}-1\right)+\frac{\gamma^2}{2a}\right]
 +o(G^{-1})
\label{eq:quad-derivative-leading}
\end{equation}
uniformly on compact partial-wetting intervals, where $\theta_0$ stays strictly between zero and $\pi$. Write $\beta=Ga_Y/\gamma$, with $a_Y$ the flat Young half-width from \eqref{eq:nl-scales}; the limit holds with $a_Y$ and the surface coefficients fixed. The common leading logarithmic term cancels between the normal and tangential loads.

This cancellation depends on the common-reference-stress loading. For arbitrary smooth amplitudes $q(a),f(a)$ in \eqref{eq:quad-loads}, with $p=q/a$ and the same fixed stiffnesses \eqref{eq:quad-stiffness}, whose coefficients are independent of $x$, the leading derivative is $r_G'=-(\log\beta)\,\partial_a(q^2+f^2)/(2\pi G)+O(G^{-1})$. Hence the logarithmic term disappears when $\partial_a(q^2+f^2)=0$, as here; omitting the tangential capillary coupling generally destroys that property. Unequal wet/dry reference stresses also change the surface stiffness operator, so the faster stationary rate below is not asserted for that more general model.

Let $I$ be the fixed compact interval of admissible half-widths containing $a_Y$ in its interior, as in Appendix~\ref{sec:nonlinear-continuum}. For fixed $K_s/\gamma,\Upsilon_0/\gamma$, these formulas and their bounded kernel derivatives imply
\begin{equation}
 \sup_{a\in I}\frac{|r_G|}{\gamma a_Y}=O\!\left(\frac{\log\beta}{\beta}\right),\qquad
 \sup_{a\in I}\left(\frac{|r_G'|}{\gamma}
              +\frac{a_Y|r_G''|}{\gamma}\right)=O(\beta^{-1}).
\label{eq:quad-C2-control}
\end{equation}
For example, $\partial_a(J_n-J_t)=(D_n-D_t)/a$ and $\partial_aD_j=-\epsilon_sD_j'/a$ are bounded on $I$. These estimates control $r_G''$ without assuming that the moving Dirac load is differentiable as a bounded functional on the energy space.

\begin{proposition}[Stationary recovery for common wet/dry mechanical coefficients]
\label{prop:quad-stationary}
For the quadratic model \eqref{eq:quad-energy} with common wet/dry mechanical coefficients and sufficiently large $\beta$, $\mathcal E_{{\rm red},G}^{\rm q}$ has a unique stationary point $a_G$ in the interior of $I$, where its derivative vanishes. It is its strict global minimum on $I$. With $\Delta\gamma=\gamma_{0,SG}-\gamma_{0,SL}$ the reference surface-energy difference and $\theta_Y=\arccos(\Delta\gamma/\gamma)$ the flat Young angle,
\begin{equation}
 \frac{|a_G-a_Y|}{a_Y}+|\theta_0(a_G)-\theta_Y|=O(\beta^{-1}),\qquad
 \Delta\gamma-\gamma\cos\theta_0(a_G)=\tfrac12r_G'(a_G).
\label{eq:quad-stationary-result}
\end{equation}
\end{proposition}
\begin{proof}
The flat energy is uniformly strictly convex by \eqref{eq:nl-flat-convexity}, and its derivative has opposite strict signs at the endpoints of $I$. The bounds through second derivatives in \eqref{eq:quad-C2-control} preserve both properties for large $\beta$. The mean-value theorem applied to $\mathcal E_W'(a_G)=-r_G'(a_G)$ gives the rate. The factor $1/2$ in the second identity comes from varying two contact lines symmetrically.
\end{proof}

\subsection{Distinction among reference-cap, spatial-cap, and microscopic angles}

The angle $\theta_0(a_G)$ in \eqref{eq:quad-stationary-result} parameterises the \emph{flat reference cap}. The spatial circular-segment angle follows instead from the displaced endpoints and segment area. Denote their first-order approximations by $b^{(1)}$ and $S^{(1)}$, where the superscript $(1)$ denotes first order in displacement, and write $\delta S=S^{(1)}-V$. The solved displacement gives
\begin{equation}
 u(a)=\frac{fJ_t}{\pi G},\qquad
 \delta S=\int_{-a}^{a}h\dd x-2ah(a)=-\frac{2aqJ_n}{\pi G},\qquad
 b^{(1)}=a+u(a),\quad S^{(1)}=V+\delta S.
\label{eq:quad-geometry-reconstruction}
\end{equation}
Define the reconstructed spatial angle $\Theta^{(1)}=\Theta(b^{(1)},S^{(1)})$ where these arguments are positive, and write $\Theta_G^{(1)}$ for its value at $a_G$. This reconstruction evaluates the exact cap function on first-order geometry; it does not solve the nonlinear equilibrium problem. Equations \eqref{eq:nl-cap-derivatives} and \eqref{eq:quad-kernel-asymptotics} give
\begin{equation}
 \Theta^{(1)}_G-\theta_Y
 =-\frac{\gamma\sin\theta_Y}{2\pi G a_Y}\log\beta+O(\beta^{-1}).
\label{eq:quad-spatial-angle-rate}
\end{equation}
Indeed, with $\delta b=b^{(1)}-a$ and $\delta\Theta$ the linear change in the cap angle, $\delta\Theta=(\delta S-2aQ\delta b)/(a^2Q')$. Here $Q=Q(\theta_0)$ is the dimensionless segment-area function in \eqref{eq:nl-cap-energy}, and $Q'$ its angle derivative; $(q+Qf)/Q'=\gamma\sin\theta_0/2$. Thus the reconstructed spatial angle has a logarithmic leading correction, unlike the $O(\beta^{-1})$ reference-angle rate. Both use the fixed horizontal as geometric reference; neither is the microscopic angle measured against a one-sided solid tangent at the contact line.

\begin{table}[tbp]
\centering
\small
\caption{Uncalibrated quadratic-continuum benchmark. In consistent dimensionless units, $\gamma=a_Y=1$, $\theta_Y=70^\circ$, $K_s=10$, $\Upsilon_0=12$, and $V=Q(70^\circ)$. The spatial column is the first-order reconstruction, not a computed nonlinear solution.}
\label{tab:quad-numerics}
\vspace{6pt}
\begin{tabular}{@{}rrrr@{}}
\toprule
$\beta$ & $a_G/a_Y$ & $\theta_0(a_G)$ (degrees) & $\Theta_G^{(1)}$ (degrees)\\
\midrule
30 & 1.00365372 & 69.661345 & 69.292118\\
100 & 1.00115544 & 69.892744 & 69.685906\\
300 & 1.00038823 & 69.963945 & 69.864093\\
1000 & 1.00011659 & 69.989170 & 69.948950\\
3000 & 1.00003887 & 69.996390 & 69.979852\\
10000 & 1.00001166 & 69.998917 & 69.992925\\
\bottomrule
\end{tabular}
\end{table}

\subsection{Finite local jumps and a nonvanishing reaction at the droplet scale}

Let $\vct\sigma$ be the incremental bulk Cauchy stress tensor, and let $\sigma_i=(\vct\sigma\vct e_z)_i$, $i\in\{x,z\}$, denote its tangential and normal traction components at $z=0$, where $\vct e_z$ is the upward unit normal. On the surface write $(v_x,v_z)=(u,h)$, and let $|D|$ denote the operator with Fourier multiplier $|k|$. The traction equations are
\begin{equation}
 \sigma_x-K_su''=t_x,\qquad \sigma_z-\Upsilon_0h''=t_z,\qquad
 \sigma_i=2G|D|v_i,\quad(v_x,v_z)=(u,h),
\label{eq:quad-traction-equations}
\end{equation}
At fixed finite $G$, the bulk tractions have locally integrable logarithmic singularities, not Dirac masses. The derivative jumps at the right contact line $x=a$, with $[v']_a=v'(a+)-v'(a-)$, therefore satisfy
\begin{equation}
 [u']_a=-f/K_s,\qquad [h']_a=-q/\Upsilon_0.
\label{eq:quad-slope-jumps}
\end{equation}
These jumps need not vanish as $G\to\infty$, even though the displacements vanish. For the surface coefficients $s_x=K_s$ and $s_z=\Upsilon_0$,
\begin{equation}
 \widehat\sigma_i(k)=\frac{2G}{2G+s_i|k|}\widehat t_i(k)
 \longrightarrow\widehat t_i(k),\qquad
 s_x=K_s,\quad s_z=\Upsilon_0,
\label{eq:quad-reaction-limit}
\end{equation}
Thus $\sigma_i\to t_i$ in tempered distributions, meaning convergence after pairing with any fixed smooth rapidly decaying test function, also along $a_G\to a_Y$. In particular this limit resolves loads on the fixed droplet scale $a_Y$, while $s_i/(2G)$ tends to zero. The elastic restoring reaction on the interfacial degrees of freedom is $-t_i$. At the right line in the Young limit its tangential resultant is $-f=\Delta\gamma$ for this common-response law, and its normal resultant is $-q=-\gamma\sin\theta_Y$, consistently with \eqref{eq:normal-reaction} and \cref{sec:young}.

More precisely, fix $a$ and a smooth function $\varphi$ supported in $(-1,1)$ with $\varphi(0)=1$. The cutoff $\varphi_\rho(x)=\varphi((x-a)/\rho)$ samples a neighbourhood of radius $\rho>0$ about the right line. With $t_i^{\rm line}$ the coefficient of the corresponding Dirac mass,
\begin{equation}
 \lim_{G\to\infty}\lim_{\rho\downarrow0}\langle\sigma_i,\varphi_\rho\rangle=0,
 \qquad
 \lim_{\rho\downarrow0}\lim_{G\to\infty}\langle\sigma_i,\varphi_\rho\rangle=t_i^{\rm line},
 \quad(t_x^{\rm line},t_z^{\rm line})=(f,q).
\label{eq:quad-noncommuting-limits}
\end{equation}
Thus the shrinking-neighborhood and rigidity limits do not commute in this model. Also $-s_iv_i''\to0$ distributionally: on any fixed observation scale as $G\to\infty$, the Dirac part due to the derivative jump is cancelled by the regular part of the surface curvature or stretch-gradient distribution. A local derivative jump alone therefore does not determine the limiting load distribution.

Positive surface moduli regularise displacement, but not every bulk strain. The extension in \cref{app:spectral} has $\partial_zv_x|_{z=0}=2|D|u-\partial_xh$, and the tangential line load gives a logarithmic singularity in $|D|u$ when $f\ne0$. The point-load calculation therefore supplies no uniform approximation of a finite-strain ridge in the immediate contact-line neighbourhood. Its reaction limit is also distinct from Appendix~\ref{sec:nonlinear-continuum}, which proves geometric and energetic convergence without identifying a limiting stress distribution.

\subsection{Bulk extension and normalisation}
\label{app:spectral}

For a nonzero Fourier mode, use the local notation $\kappa=|k|$ for the wave-number magnitude and $U=\widehat u(k)$, $H=\widehat h(k)$ for the surface-displacement amplitudes. Thus $\kappa$ and $H$ here do not denote a curvature or the finite-layer depth of Appendix~\ref{sec:nonlinear-continuum}. With $\pi$ the incremental pressure enforcing incompressibility, the decaying bulk displacement and pressure are
\begin{align}
 \widehat v_x(k,z)&=[U+(\kappa U-ikH)z]e^{\kappa z},\nonumber\\
 \widehat v_z(k,z)&=[H+(-ikU-\kappa H)z]e^{\kappa z},\nonumber\\
 \widehat\pi(k,z)&=(-2iGkU-2G\kappa H)e^{\kappa z}.
\label{eq:quad-bulk-extension}
\end{align}
Direct substitution gives $\nabla\cdot\vct v=0$, $G\Delta\vct v-\nabla\pi=0$, and $\widehat{\vct\sigma\vct e_z}=2G|k|(U,H)$ at $z=0$, where $\nabla$ and $\Delta$ are the gradient and Laplacian in $(x,z)$. The bulk energy is one-half of boundary traction work, which fixes both the coefficient in \eqref{eq:quad-stiffness} and the $1/(\pi G)$ factor in \eqref{eq:quad-elimination}; in this factor $\pi$ is the usual numerical constant. The loads transform as
\begin{equation}
 \widehat t_x=-2if\sin(ka),\qquad
 \widehat t_z=2q\left[\cos(ka)-\frac{\sin(ka)}{ka}\right].
\label{eq:quad-load-transforms}
\end{equation}
They vanish at zero wave number with orders one and two, respectively. Together with $D_i\sim2G|k|$ at small $|k|$ and $D_i\sim s_i k^2$ at large $|k|$, this proves that $\int|\widehat t_i|^2/D_i\dd k$ is finite. Weighted Cauchy--Schwarz makes the loads continuous on the energy space, and completion of the square gives its unique minimising class. The explicit Fourier solutions select decaying continuous surface displacements; their bulk extensions need not be square-integrable over the infinite half-space.

\subsection{Positive integral representation and exact kernel reduction}

In this subsection only, $x>0$ denotes a dimensionless kernel argument, not the horizontal position, and $s,y$ are dimensionless integration variables. Define the auxiliary function $A$ by
\begin{equation}
 A(x)=\int_0^\infty\frac{1-\cos(xs)}{s(1+s)}\dd s
 =\frac12\int_0^\infty e^{-s}\log(1+x^2/s^2)\dd s.
\label{eq:quad-A-positive}
\end{equation}
The second representation follows by writing $(1+s)^{-1}$ as a Laplace integral and integrating the difference of cosines. With $R(x)$ denoting the nonnegative remainder after subtracting $\log x+\gamma_E$, it yields
\begin{equation}
 \begin{gathered}
 A(x)=\log x+\gamma_E+R(x),\qquad
 0\le R(x)=\tfrac12\int_0^\infty e^{-s}\log(1+s^2/x^2)\dd s\le x^{-2},\\
 xA'(x)=\int_0^\infty e^{-s}\frac{x^2}{s^2+x^2}\dd s.
 \end{gathered}
\label{eq:quad-A-estimates}
\end{equation}
Primes on $A$ now denote derivatives with respect to its argument $x$. With $x=2/\epsilon$, elementary integration gives
\begin{equation}
 J_t(\epsilon)=\tfrac12 A(x),\qquad
 J_n(\epsilon)=\frac{2}{x^2}\int_0^x yA(y)\dd y-\tfrac12 A(x),
 \qquad D_t=\tfrac12xA',\quad D_n=A-2J_n-D_t.
\label{eq:quad-A-kernels}
\end{equation}
For example, integrate $y[1-\cos(ys)]$ first in the second formula. Adding $J_t$ gives the numerator
$1-\sin(2t)/t+\sin^2t/t^2$, which is exactly
$(\cos t-\sin t/t)^2+\sin^2t$. This proves the normal-kernel reduction without separating individually divergent integrals.

The bound on $R$ proves \eqref{eq:quad-kernel-asymptotics}: split $\int_0^x yR(y)\dd y$ at one; it is $O(1+\log x)$ for large $x$. Also $0<xA'<1$ and $1-xA'\le2/x^2$. Define the tangential and normal kernel numerators $\phi_t(t)=\sin t$ and $\phi_n(t)=\cos t-\sin t/t$. They give the absolutely convergent identities
\begin{equation}
 D_j(\epsilon)=\int_0^\infty\frac{\epsilon\phi_j(t)^2}{(1+\epsilon t)^2}\dd t,
 \qquad
 \epsilon D_j'(\epsilon)=\int_0^\infty
 \frac{\epsilon(1-\epsilon t)\phi_j(t)^2}{(1+\epsilon t)^3}\dd t.
\label{eq:quad-D-bounds}
\end{equation}
Because $|\phi_t|\le1$ and $|\phi_n|\le2$, both integrals are uniformly bounded in absolute value by $1$ and $4$, respectively. The explicit formulas and the remainder estimate give $D_j\to1/2$.

For numerical evaluation, define the sine integral $\mathrm{Si}(x)=\int_0^x\sin t\,\dd t/t$ and cosine integral $\mathrm{Ci}(x)=-\int_x^\infty\cos t\,\dd t/t$, the latter as an improper integral. Then
\begin{align}
 A(x)&=\log x+\gamma_E-\cos x\,\mathrm{Ci}(x)
                    -\sin x\,[\mathrm{Si}(x)-\pi/2],\nonumber\\
 A'(x)&=\sin x\,\mathrm{Ci}(x)-\cos x\,[\mathrm{Si}(x)-\pi/2].
\label{eq:quad-A-special}
\end{align}
Equivalently, the nonoscillatory representation \eqref{eq:quad-A-estimates} avoids cancellation at large $x$. The identity $A''+A=\log x+\gamma_E$ gives
\begin{equation}
 J_n=\log x+\gamma_E-\tfrac12-\frac{2A'}x+\frac{2A}{x^2}-\tfrac12A.
\label{eq:quad-Jn-evaluation}
\end{equation}
For small $x$, evaluating the integral in \eqref{eq:quad-A-kernels} avoids subtraction of large nearly equal terms. The kernel representations, reduced-energy derivative, and geometric reconstruction were checked independently by symbolic substitution, positive quadrature, and finite differences.

\section{Technical variational details}
\label{sec:technical-variations}
\label{sec:technical-details}

\subsection{Finite-strain surface stress and thermodynamic controls}
\label{sec:surface-stress-exact}

Consider a local first-gradient material-surface energy, meaning that it depends on local deformation but not on curvature or deformation gradients along the surface; no independent couple stress is retained. The subscript $s$ denotes one specified material surface. In orthonormal reference tangent coordinates, its surface deformation gradient $\vct F_s$ is a rank-two $3\times2$ map, $\vct C_s=\vct F_s^{\mathsf T}\vct F_s$ is the surface metric, and $J_s=\sqrt{\det\vct C_s}>0$ is the current-to-reference area ratio.
Write the reference-area potential as $\Psi_s=J_s\gamma_s(\vct C_s,\vct q_s)$, where $\gamma_s$ is per current area and $\vct q_s$ collects referential internal variables. For this objective, or frame-indifferent, potential, differentiation at fixed thermodynamic controls and $\vct q_s$ gives the surface first Piola stress $\vct P_s=\partial_{\vct F_s}\Psi_s$:
\[
\vct P_s
=J_s\gamma_s\vct F_s\vct C_s^{-1}
+2J_s\vct F_s\frac{\partial\gamma_s}{\partial\vct C_s}.
\]
Here $\partial_{\vct F_s}J_s=J_s\vct F_s\vct C_s^{-1}$. The push-forward in \eqref{eq:surface-piola-pushforward} therefore gives the current surface stress $\vct\Upsilon_s$:
\begin{equation}
\vct\Upsilon_s=\gamma_s\vct P_{\Sigma_s}
+2\vct F_s\frac{\partial\gamma_s}{\partial\vct C_s}\vct F_s^{\mathsf T},
\qquad
\vct P_{\Sigma_s}=\vct F_s\vct C_s^{-1}\vct F_s^{\mathsf T}.
\label{eq:finite-shuttleworth}
\end{equation}
Here $\vct P_{\Sigma_s}$ is the orthogonal projector onto the current surface tangent plane, distinct from the Piola stress $\vct P_s$. In particular, strain-independent $\gamma_s$ gives $\vct\Upsilon_s=\gamma_s\vct P_{\Sigma_s}$. These specified reference-area, Piola and current-stress measures give the finite-strain relation \cite{GurtinMurdoch1975,HeydenBain2024}.

Let $\psi_s^H(\vct F_s,\{n_i^0\},T)$ be a Helmholtz energy per reference area, where $n_i^0$ is the independently exchangeable adsorbed amount of species $i$ per reference area. Here $T$ denotes temperature, not the time interval in Appendix~\ref{sec:transport-estimate}. At fixed adsorbate chemical potentials $\mu_i$, the equilibrium reduced potential is
\begin{equation}
\Psi_s(\vct F_s,T,\{\mu_i\})
=\left[\psi_s^H-\sum_i\mu_i n_i^0\right]_{n_i^0=n_{i,*}^0},
\qquad
\frac{\partial\psi_s^H}{\partial n_i^0}=\mu_i,
\qquad \gamma_s=\frac{\Psi_s}{J_s}.
\label{eq:surface-grand-potential}
\end{equation}
The subscript $*$ denotes the equilibrium adsorbed amounts. On a differentiable stable branch, the envelope identity gives
$\partial_{\vct F_s}\Psi_s=\partial_{\vct F_s}\psi_s^H$ evaluated at equilibrium.

For differentiation at fixed adsorbed amount, let $A$ be the current interface area, $N$ the adsorbed amount, $c=N/A$ its current concentration and $f(c)$ the Helmholtz energy per current area. Direct differentiation of $Af(N/A)$ at fixed $N$ gives $f-cf'$, where the prime denotes differentiation with respect to $c$ \cite{Olives2010}.

With the isotropic stress convention $\vct S_i^s=-p_i^s\vct P_i$, where $i$ identifies an interface, $\vct P_i$ is its tangent projector and $p_i^s=-\gamma_i$ is its surface pressure, substitution into Young's relation gives
\begin{equation}
p_{LG}^s\cos\theta=p_{SG}^s-p_{SL}^s.
\label{eq:young-surface-pressure}
\end{equation}

\subsection{Smooth migration and the reference-to-current measures}
\label{sec:smooth-migration}

Let the reference wet--dry boundary be $\mathcal C$, with unit tangent $\vct T_{\mathcal C}$ and unit reference co-normal $\vct N_{\mathcal C}$ pointing from $SL$ toward $SG$. Its current image is $\Gamma$, and $\vct n_\Gamma^S$ is the unit co-normal within the current solid surface, oriented toward increasing wet area. Co-normals are tangent to the surface and perpendicular to the boundary. A tangential reference perturbation $\vct X^\epsilon$ of a boundary point $\vct X$, parameterised by $\epsilon$, has signed reference co-normal displacement $\eta_0$:
\begin{equation}
\left.\frac{d\vct X^\epsilon}{d\epsilon}\right|_{\epsilon=0}
=\eta_0\vct N_{\mathcal C}.
\label{eq:config-variation}
\end{equation}
Assume a fixed smooth placement with the same boundary value of the surface deformation gradient on both sides of $\mathcal C$. Define the line stretch $j_\Gamma=|\vct F_s\vct T_{\mathcal C}|$ and the co-normal stretch factor $m_\Gamma=J_s/j_\Gamma$. The current co-normal component of $\vct F_s\vct N_{\mathcal C}$ equals $m_\Gamma$, since the transformed tangent and co-normal span a parallelogram of area $J_s$. Thus the signed current co-normal displacement $\eta$ and the current and reference line elements $\dd\ell$, $\dd\ell_0$ satisfy
\begin{equation}
\eta=m_\Gamma\eta_0,
\qquad \dd\ell=j_\Gamma\dd\ell_0,
\qquad \eta\,\dd\ell=J_s\eta_0\,\dd\ell_0.
\label{eq:migration-measure}
\end{equation}
Any current line-tangential component only reparametrizes the line and has zero pairing with the liquid--gas co-normal.

Let $\vct N_{LG}$ be the liquid--gas unit co-normal directed from the line into that interface. For the local potentials above, without retained line energy or an explicitly partition-dependent bulk/loading term, the migration contribution $\delta_a\mathcal E$ to the energy variation is
\begin{align}
\delta_a\mathcal E\big|_\Gamma
&=\int_{\mathcal C}
\left[\Psi_{SL}-\Psi_{SG}
-j_\Gamma\gamma_{LG}\vct N_{LG}\cdot\vct F_s\vct N_{\mathcal C}\right]
\eta_0\dd\ell_0 \nonumber\\
&=\int_\Gamma
\left[\gamma_{SL}-\gamma_{SG}
-\gamma_{LG}\vct N_{LG}\cdot\vct n_\Gamma^S\right]
\eta\dd\ell.
\label{eq:smooth-migration-variation}
\end{align}
The factor $j_\Gamma$ is required because the liquid--gas edge traction is per current line length. Its interior variation vanishes after constrained interface stationarity; a compatible extension of the endpoint variation is understood. With the sign convention of \eqref{eq:config-force},
\begin{equation}
f_\Gamma^{\Conf}
=\gamma_{SG}-\gamma_{SL}
+\gamma_{LG}\vct N_{LG}\cdot\vct n_\Gamma^S.
\label{eq:smooth-config-explicit}
\end{equation}
Each surface state uses its own equilibrium potential at the same prescribed deformation. On the flat wall, $\vct N_{LG}\cdot\vct n_\Gamma^S=-\cos\theta$.

The corresponding power pairing uses the line speed relative to solid material. If $\vct w_\Gamma$ is the geometric line velocity and $\vct v_S$ the common solid velocity trace, then
\begin{equation}
V_\Gamma^S=(\vct w_\Gamma-\vct v_S)\cdot\vct n_\Gamma^S.
\label{eq:relative-line-speed}
\end{equation}
At a ridge, the common-trace assumption fails; the compatible one-sided variations in \eqref{eq:corner-compatibility} replace the smooth-trace calculation. Inserting an arbitrary one-sided tip stretch into \eqref{eq:smooth-migration-variation} is not justified.

\subsection{Eliminated fields, contour resultants and corner boundary terms}

For each nearby partition position $a$, suppose the remaining fields $y$ lie on a differentiable stationary branch $y_*(a)$. These fields include the solid placement, liquid--gas shape and any equilibrated surface variables. Use a common-domain parametrization and the augmented functional $\mathcal L_{\rm aug}=\mathcal E+\sum_j\Lambda_j C_j(y,a)$, where $C_j=0$ are the attachment, volume and other constraints and $\Lambda_j$ their Lagrange multipliers. Denote the stationary collection of multipliers by $\Lambda_*(a)$. Differentiating $C_j(y_*(a),a)=0$ and using stationarity in $y$ gives
\begin{equation}
\frac{d\mathcal E_{\rm red}}{da}
=\left.\frac{\partial\mathcal L_{\rm aug}}{\partial a}
\right|_{y=y_*(a),\,\Lambda=\Lambda_*(a)},
\qquad \mathcal E_{\rm red}(a)=\mathcal E(y_*(a),a).
\label{eq:envelope-config}
\end{equation}
An explicitly $a$-dependent constraint contributes $\Lambda_j\partial_a C_j$. If constraints have been eliminated in an $a$-independent parametrization, the derivative reduces to $\partial_a\mathcal E$.

For the mechanical contour integral, let $C_\rho$ bound a current cross-sectional region of radius $\rho$ around a locally straight line, with $\vct n_c$ outward from that region. Write $\vct T_S$ for the solid Cauchy stress, $\dd s$ for current contour arclength, and $\vct f_\Gamma^{\rm ext}$ for any additional applied force per current line length. With bounded fluid bulk tractions, the line force balance has the form
\begin{equation}
\lim_{\rho\downarrow0}\int_{C_\rho\cap S}\vct T_S\vct n_c\,\dd s
+\vct\Upsilon_{SL}\vct N_{SL}
+\vct\Upsilon_{SG}\vct N_{SG}
+\gamma_{LG}\vct N_{LG}
+\vct f_\Gamma^{\rm ext}=\vct0.
\label{eq:mechanical-contour-balance}
\end{equation}
All sheet co-normals point from the line into their sheets. Singular fluid resultants or curvature forces from a true line stress, when retained, require extra terms.

For a bulk configurational contour integral, define the Eshelby tensor
\begin{equation}
\vct{\mathbb E}=W_S\vct I-\vct F_S^{\mathsf T}\vct P_S,
\label{eq:eshelby}
\end{equation}
Here $W_S$ is the bulk stored energy per reference volume, $\vct F_S$ the bulk deformation gradient, $\vct I$ the identity tensor and $\vct P_S$ the bulk first Piola stress, including the incompressibility multiplier where appropriate \cite{Eshelby1951,Gurtin1995,Gurtin2000}. The tensor $\vct{\mathbb E}$ maps reference directions to configurational force per reference area. Its contour and force directions are therefore referential.

The mechanical and configurational contour limits require separate estimates. In two dimensions, a singularity bound $|\vct T_S|=O(\rho^{-q})$ with exponent $q<1$ makes the mechanical resultant vanish as $O(\rho^{1-q})$. An analogous reference-coordinate bound on $\vct{\mathbb E}$ makes its resultant vanish. The stress bound alone does not control $\vct{\mathbb E}$ if $\vct F_S$ is unbounded. For an incompressible neo-Hookean solid, bounded $\vct F_S$ and logarithmic stress give both vanishing resultants, as in the corresponding nonlinear ridge class \cite{Pandey2020}.

The corner boundary calculation starts from the surface energy per unit unchanged transverse length,
\begin{equation}
\mathcal E_\Sigma
=\int_{X_-}^{a_0}\psi_{SL}(\lambda_{SL})\dd X
+\int_{a_0}^{X_+}\psi_{SG}(\lambda_{SG})\dd X
+\gamma_{LG}A_{LG}.
\label{eq:corner-energy}
\end{equation}
Here $X$ is reference surface arclength, $X_-$ and $X_+$ are fixed endpoints, $a_0$ is the material contact-line coordinate, and $A_{LG}$ is liquid--gas area per unit transverse length. The one-sided stretches and reference-area potentials are those in \eqref{eq:corner-conjugates}. Integration by parts gives wet and dry endpoint terms $\Upsilon_{SL}\vct t_{SL}\cdot\delta\vct r_{SL}$ and $-\Upsilon_{SG}\vct t_{SG}\cdot\delta\vct r_{SG}$, plus $(\psi_{SL}-\psi_{SG})\delta a_0$. Substituting \eqref{eq:corner-compatibility} and the liquid--gas endpoint term $-\gamma_{LG}\vct N_{LG}\cdot\delta\vct r_\Gamma$ gives \eqref{eq:corner-first-variation}. The selected bulk and loading variations must be added to this surface boundary identity.

Under a homogeneous relabeling $\widetilde X=cX$ with positive constant $c$, $\widetilde\lambda=\lambda/c$ and $\widetilde\psi(\widetilde\lambda)=\psi(c\widetilde\lambda)/c$. Tildes denote quantities expressed in the new reference coordinate. Consequently $\widetilde\Upsilon=\Upsilon$ and $\widetilde H=H/c$, preserving $H\,\delta a_0$.

\subsection{Evaluation of the illustrative surface law}

For \eqref{eq:benchmark-law}, direct differentiation gives
\begin{align}
\gamma_{S\alpha}(\lambda_s)
&=\gamma_{0,S\alpha}+b_{S\alpha}(1-\lambda_s^{-1})
+\frac{k_{S\alpha}}{2\lambda_s}(\lambda_s-1)^2,
\label{eq:benchmark-gamma}\\
\Upsilon_{S\alpha}(\lambda_s)
&=\gamma_{0,S\alpha}+b_{S\alpha}+k_{S\alpha}(\lambda_s-1),
\label{eq:benchmark-upsilon}\\
H_{S\alpha}(\lambda_s)
&=-b_{S\alpha}+\frac{k_{S\alpha}}2(1-\lambda_s^2).
\label{eq:benchmark-H}
\end{align}
The Hessian of $\psi(|\partial_X\vct r|)$ has longitudinal eigenvalue $k_{S\alpha}$ and transverse eigenvalue $\Upsilon_{S\alpha}/\lambda_s$. Thus tensile membrane behaviour requires $\Upsilon_{S\alpha}>0$ on the chosen stretch interval; $k_{S\alpha}\ge0$ alone does not exclude compression instability.

For an unpinned corner with zero bulk configurational resultant, \eqref{eq:corner-nopinning} becomes
\begin{equation}
b_{SL}+\frac{k_{SL}}2(\lambda_{SL}^2-1)
=b_{SG}+\frac{k_{SG}}2(\lambda_{SG}^2-1).
\label{eq:benchmark-nopinning}
\end{equation}
For example, $b_{SL}=b_{SG}=0$, $k_{SG}=2k_{SL}>0$ and $\lambda_{SL}=5/4$ require $\lambda_{SG}=\sqrt{41/32}$. This is a local compatibility calculation, not a bulk ridge solution or an identification of either tip stretch with the far-field stretch.

On the separately prescribed common flat placement, the effective tangential reaction at Young equilibrium is instead
\begin{equation}
R_t^{\rm eff}
=\frac{H_{SG}(\lambda_\infty)-H_{SL}(\lambda_\infty)}{\lambda_\infty}
=\frac{b_{SL}-b_{SG}}{\lambda_\infty}
+\frac{k_{SG}-k_{SL}}{2\lambda_\infty}(1-\lambda_\infty^2).
\label{eq:benchmark-reaction}
\end{equation}
The numerical values in \eqref{eq:benchmark-numbers} follow from $\lambda_\infty=5/4$, $b_{SL}=b_{SG}=0$, $\gamma_{0,SL}=40$, $\gamma_{0,SG}=50$, $k_{SL}=20$, $k_{SG}=40$ and $\gamma_{LG}=30$, with all energetic coefficients in $\mathrm{mN\,m^{-1}}$. The mechanical check is $-4.5+60-45-30(0.35)=0$; the configurational check is $51-40.5-30(0.35)=0$. At $\lambda_\infty=1$, $R_t^{\rm eff}=b_{SL}-b_{SG}$, so a tangential reaction does not require macroscopic prestretch.

\subsection{Derivative control and the fixed-wall endpoint variation}

In nondimensional units, let $\vct y$ be a vector of displacement coordinates, $\vct b(a)$ a differentiable loading function, $K>0$ a stiffness parameter and $\mathcal E_W(a)$ the fixed-wall energy. Consider the quadratic penalty energy
\begin{equation}
\mathcal E_K(\vct y,a)
=\frac K2|\vct y|^2-\vct b(a)\cdot\vct y+\mathcal E_W(a),
\qquad K>0.
\label{eq:penalty-energy}
\end{equation}
Exact elimination gives
\begin{equation}
\vct y_*=\frac{\vct b}{K},\qquad
\mathcal E_{{\rm red},K}=\mathcal E_W-\frac{|\vct b|^2}{2K},\qquad
\mathcal E_{{\rm red},K}'=\mathcal E_W'-\frac{\vct b\cdot\vct b'}{K}.
\label{eq:penalty-elimination}
\end{equation}
Here $\vct y_*$ minimises the energy at fixed $a$, and primes denote derivatives with respect to $a$. Uniform bounds on $\vct b$ and $\vct b'$ give derivative convergence at rate $K^{-1}$, while the restoring reaction $-K\vct y_*=-\vct b$ remains finite. The parameter $K$ belongs only to this nondimensional illustration and is distinct from the material surface moduli.

Uniform energy convergence alone does not imply convergence of its derivative. For dimensionless $a$ and a positive parameter $\epsilon\to0$, the functions $r_\epsilon(a)=\epsilon\sin(a/\epsilon)$ satisfy $|r_\epsilon(a)|\le\epsilon$ for every $a$, whereas $r_\epsilon'(a)=\cos(a/\epsilon)$ does not converge uniformly to zero on any interval of positive length. This counterexample concerns the implication between two mathematical convergence statements; it is not a constitutive model.

For a stationary continuum family, identify smooth migration fields on the varying and limiting surfaces. Assume fixed anchoring, convergent geometry on the fixed macroscopic scale $L$ chosen in the main text, convergent equilibrium surface states, reversible migration and no retained line energy. This geometric convergence concerns distances comparable with $L$ and large compared with the shrinking elastocapillary length $\ell_{\ec}$. Let $\Gamma_W$ be the limiting fixed-wall line, $\gamma_{S\alpha}^*$ the limiting surface potentials and $\theta_*$ the limiting contact angle on that scale. Define $\mathcal R_{\epsilon_{\ec}}[\eta]$ as the difference between the family first variation and the fixed-wall target first variation, with $\epsilon_{\ec}=\ell_{\ec}/L$ as in \eqref{eq:two-small-parameters}. Then, by definition and stationarity,
\begin{equation}
\begin{aligned}
0=\delta_a\mathcal E_{\rm red}^{\epsilon_{\ec}}[\eta]
={}&\int_{\Gamma_W}
\bigl(\gamma_{SL}^*-\gamma_{SG}^*+\gamma_{LG}\cos\theta_*\bigr)
\eta\dd\ell
+\mathcal R_{\epsilon_{\ec}}[\eta].
\end{aligned}
\label{eq:limit-variation-residual}
\end{equation}
This identity asserts no residual estimate. If $\mathcal R_{\epsilon_{\ec}}[\eta]\to0$ for every admissible test field, the limit satisfies Young's condition. Equation \eqref{eq:first-variation-convergence} is one sufficient reduced-coordinate hypothesis. If the residual instead converges to $\int_{\Gamma_W}r_*\eta\dd\ell$ for a nonzero line-force density $r_*$, the limiting coefficient contains $+r_*$; singular residuals require a weak formulation. A passage of global minimality is an alternative to derivative convergence, as in Appendix~\ref{sec:nonlinear-continuum}; it does not justify an arbitrary stationary branch.

Finally, the fixed-wall target itself has a direct constrained variation. Omit gravity, let $\vct n_\Sigma$ be the liquid--gas unit normal directed out of the liquid, and define $\kappa=\operatorname{div}_\Sigma\vct n_\Sigma$ as the sum of principal curvatures, positive for a spherical liquid drop. Let $\zeta$ be the normal virtual displacement and $V_L$ the liquid volume. The capillary surface potential is $\mathcal E_s=\gamma_{LG}A_{LG}+\gamma_{SL}A_{SL}+\gamma_{SG}A_{SG}$, where $A_i$ denotes the current area of interface $i$. Its volume multiplier $\Delta p$ is the liquid-minus-gas pressure difference at equilibrium. With the wall fixed,
\begin{equation}
\begin{aligned}
\delta(\mathcal E_s-\Delta p\,V_L)
={}&\int_{\Sigma_{LG}}(\gamma_{LG}\kappa-\Delta p)\zeta\dd A\\
&+\int_\Gamma
(\gamma_{SL}-\gamma_{SG}+\gamma_{LG}\cos\theta)\eta\dd\ell.
\end{aligned}
\label{eq:young-full-first-variation}
\end{equation}
Here $\delta V_L=\int_{\Sigma_{LG}}\zeta\dd A$; there is no additional line measure for a nondegenerate contact line moving tangentially on the fixed wall. Independent interior and endpoint variations give Young--Laplace and Young equilibrium.

\end{document}